\documentclass[final,twoside,11pt]{entics} 
\usepackage{macros}
\usepackage{enticsmacro}
\usepackage{graphicx}
\def\conf{MFPS 2022} 	
\volume{1}			
\def\volu{1}			
\def\papno{2}			
\def\mydoi{{\normalshape \href{https://doi.org/10.46298/entics.10507}{10.46298/entics.10507}}}
\def\copynames{H. DeYoung, F. Pfenning}
\def\lastname{DeYoung and Pfenning}
\def\runtitle{Data Layout from a Type-theoretic Perspective}
\begin{document}
\begin{frontmatter}
  \title{Data Layout from a Type-Theoretic Perspective} 
  \subtitle{{\large Invited Paper}}
  \author{Henry DeYoung\thanksref{myemail}}
  \author{Frank Pfenning\thanksref{coemail}\thanksref{ALL}}
  \address{Computer Science Department\\ Carnegie Mellon University\\
    Pittsburgh, United States}
  \thanks[myemail]{Email: \href{mailto:hdeyoung@cs.cmu.edu} {\texttt{\normalshape hdeyoung@cs.cmu.edu}}} 
  \thanks[coemail]{Email: \href{mailto:fp@cs.cmu.edu} {\texttt{\normalshape fp@cs.cmu.edu}}}
  \thanks[ALL]{These are notes for an invited talk by the second author given at MFPS 2022.  We would like to thank the program committee for the invitation.
We would also like to thank Sophia Roshal for her feedback on a draft of this paper.
This material is based upon work supported by the United States Air Force and DARPA under Contract No.\ FA8750-18-C-0092.}   
\begin{abstract}
  The specifics of data layout can be important for the efficiency of functional programs and interaction with external libraries.
In this paper, we develop a type-theoretic approach to data layout that could be used as a typed intermediate language in a compiler or to give a programmer more control.
Our starting point is a computational interpretation of the semi-axiomatic sequent calculus for intuitionistic logic that defines abstract notions of cells and addresses.
We refine this semantics so addresses have more structure to reflect possible alternative layouts without fundamentally departing from intuitionistic logic.
We then add recursive types and explore example programs and properties of the resulting language.
\end{abstract}
\begin{keyword}
  sequent calculus, Curry--Howard correspondence, futures-based concurrency, data layout
\end{keyword}
\end{frontmatter}
\section{Introduction}\label{sec:introduction}

The Curry-Howard isomorphism establishes a firm relationship between the propositions of
intuitionistic logic and types for functional programs.  Even if both sides of this correspondence
are intuitionistic, the precise relation between proofs and functional programs varies significantly
with the proof system.  Curry~\cite{Curry34}, for example, related Hilbert-style axiomatic proofs to
combinators and combinatory reduction.  Howard~\cite{Howard69}, on the other hand, related natural
deductions to terms in the typed $\lambda$-calculus.
As a further example, Herbelin~\cite{Herbelin94csl} related an intuitionistic sequent calculus with a
``stoup'' (LJT) to a $\lambda$-calculus with a form of explicit substitution.  The computational
mechanism in each of these examples is quite different and in each case derives from a
proof-theoretic notion of reduction.  The logical motivation for these reductions stems from the
desire to prove, constructively at the meta-level, the \emph{consistency} of specific systems of
inference rules~\cite{Gentzen35}.  Moreover, proofs that are fully reduced (called \emph{normal} or
\emph{cut-free}, depending on the system) exhibit the \emph{subformula property}, which means they
can serve as \emph{verifications}~\cite{Dummett91,MartinLof83}.

The basis for this paper is the recently discovered
\emph{semi-axiomatic sequent calculus}~\cite{DeYoung20fscd} and its correspondence to futures~\cite{Baker77sigplan,Halstead85}.  Briefly, the semi-axiomatic sequent calculus starts
from the intuitionistic sequent calculus and replaces the right rules for positive connectives
(positive conjunctions $A_1 \land A_2$, positive unit $\top$, and disjunctions $A_1 \lor A_2$) and the left rules for negative connectives (implications $A_1 \limp A_2$ and
negative conjunctions $A_1 \with A_2$) by axioms.  For example, the six axioms $A_1, A_2 \vdash A_1 \land A_2$;
$A_k \vdash A_1 \lor A_2$ for $k \in \set{1,2}$; $A_1, A_1 \limp A_2 \vdash A_2$; and $A_1 \with A_2 \vdash A_k$ for $k \in \set{1,2}$ replace the usual $\rrule{\land}$, $\rrule{\lor}_k$, $\lrule{\limp}$, and $\lrule{\with}_k$ rules, respectively,
while the $\lrule{\land}$, $\lrule{\lor}$, $\rrule{\limp}$, and $\rrule{\with}$ rules remain unchanged.  This
restructuring leads to a failure of Gentzen's cut elimination theorem~\cite{Gentzen35}: some cuts
(called \emph{snips}) are essential and cannot be eliminated.  Fortunately, all snips arise from the
new axioms and can be shown to obey a subformula property~\cite[Theorem~7]{DeYoung20fscd} derived from the
evident subformula property of the new axioms ($A_1$ and $A_2$ are subformulas of $A_1 \land A_2$, etc.).

Computationally, a process is assigned to each proof in the semi-axiomatic sequent calculus, arriving at a type theory dubbed SAX~\cite[section~5]{DeYoung20fscd}.  Both cuts and snips are treated as the
allocation of a future, with the two premises of the cut or snip computing in parallel.  The first
premise computes and writes the value of the future, while the second may read its value,
potentially blocking until it has been written.  The \emph{type} of the future (that is, the
\emph{cut formula}) determines the form of value that it ultimately holds.

The key observation underlying this paper is that we have the freedom to give different
computational interpretations to cuts and snips.  The subformula property of the new axioms,
reflected in the snips, is manifest in \emph{projections} from the values of larger types to those of smaller
types.  For example, ordinarily we might think of assigning terms to the axiom
$A_1, A_2 \vdash A_1 \land A_2$ as $x_1{:}A_1, x_2{:}A_2 \vdash \langle x_1, x_2\rangle : A_1 \land A_2$.  Instead, if we ensure that
$x : A_1 \land A_2$ has already been allocated, the axiom expresses \emph{address projections}
$x \cdot \pi_1 : A_1, x \cdot \pi_2 : A_2 \vdash x : A_1 \land A_2$.  Taking it further, this axiom can be seen as \emph{merely computing the addresses of $A_1$ and $A_2$, given
  the address for the pair $A_1 \land A_2$}.  Under this interpretation, only cuts allocate memory for a
future.  A snip is then the parallel composition of two processes that respectively write to and
read from a shared portion of a future that has already been allocated.

When we go beyond the usual propositions-as-types correspondence and add recursive types and recursive processes to our language, we have to slightly modulate our approach.  For example, a type
of lists of booleans satisfying the equation
$\boollist \defd \plus*{\mt{nil}\colon \one , \mt{cons}\colon \bool \tensor \boollist}$
would require an unbounded (or unpredictable) amount of space for a value of type $\boollist$.  In
order to avoid this problem, we introduce a type constructor $\dn A$, originating in
adjoint logic~\cite{Benton94csl,Reed09un,Pruiksma18aun},
that is
inhabited by addresses for cells of type $A$, \emph{i.e.}, pointers.  Logically, this has no force in the sense that
$A \mathrel{\dashv\,\vdash} \dn A$, but it affects the fine structure of proofs.  Every recursive type
must be \emph{guarded} by a $\dn$ shift, as in
$\boollist \defd \plus*{\mt{nil}\colon \one , \mt{cons}\colon \dn \bool \tensor \dn \boollist}$.
This expresses a layout where a binary number is either just the tag $\mt{nil}$ or a
tag $\mt{cons}$ together with a pair of addresses, $\mt{cons}\:a_1\:a_2$.  Through different ways to place $\dn$ shifts
we can control the layout of the data.
For example, here we would likely prefer $\boollist \defd \plus*{\mt{nil}\colon \one , \mt{cons}\colon \bool \tensor \dn \boollist}$, where the address of a boolean is replaced by the boolean itself.

In summary, the main contributions of this paper are:
\begin{itemize}
\item a reformulated semi-axiomatic sequent calculus that syntactically distinguishes
  cuts and snips so they can be assigned a different dynamics (\cref{sec:SNAX}); and simultaneously
\item a new type theory, derived by Curry--Howard interpretation of the reformulated semi-axiomatic sequent calculus, for specifying data layout for futures-based concurrency (also in \cref{sec:SNAX}); and
\item proofs of preservation and progress (including equirecursive
  types and recursive processes, \cref{sec:SNAX:safety}).
\end{itemize}
Related work is discussed in \cref{sec:related}.
An extended version of this paper can be found at arXiv~\cite{DeYoung22arxiv}.

\section{SAX: A semi-axiomatic type theory for shared memory concurrency}\label{sec:SAX}

The semi-axiomatic sequent calculus is a presentation of intuitionistic logic that blends inference rules of the sequent calculus with axioms of the Hilbert calculus~\cite{DeYoung20fscd}.
Perhaps surprisingly, there is a Curry--Howard correspondence between the semi-axiomatic sequent calculus and a type theory for futures~\cite{Baker77sigplan,Halstead85}, a form of write-once shared memory concurrency.
Here we give a detailed review of the semi-axiomatic sequent calculus and its SAX type theory (though slightly altered to use explicit rules of weakening and contraction) because they serve as the cornerstones for our SNAX type theory.

\paragraph*{Judgmental aspects.}
Following the usual Curry--Howard pattern, the propositions of intuitionistic logic become SAX types, sequents become SAX static typing judgments, and their proofs become SAX processes.
Sequents and SAX typing judgments correspond as follows:
\begin{equation*}
    \underbrace{A_1 , A_2 , \dotsc , A_n}_{\text{\emph{antecedents}}} \vdash \underbrace{A\vphantom{A_n}}_{\mathclap{\text{\emph{succedent}}}}
\qquad\qquad\quad
    \underbrace{a_1{:}A_1 , a_2{:}A_2 , \dotsc , a_n{:}A_n\vphantom{)}}_{\text{\emph{may read from}}} \vdash P :: \underbrace{(a : A)\vphantom{A_n}}_{\mathclap{\text{\emph{must write to}}}}
\end{equation*}
where the typing judgment is read as ``Process $P$ may read data of types $A_1, A_2, \dotsc, A_n$ from addresses $a_1, a_2, \dotsc, a_n$, respectively, and must write data of type $A$ to address $a$.''
This discipline also serves to enforce the invariant that each address is written by exactly one process.
For this reason, we will sometimes refer to the address that a process must write to as the process's \emph{destination}.

In both the semi-axiomatic proof theory and the SAX type theory, we use the metavariable $\ctx$ to stand for contexts -- of either antecedents or readable addresses, respectively.
SAX addresses have no structure, being only variables $x$ that will be mapped to concrete addresses $\alpha$ at runtime.
\begin{equation*}
  \begin{alignedat}[t]{2}
    \text{\emph{Contexts}} &\quad&
      \ctx &\Coloneqq \ctxe* \mid \ctx , A
  \end{alignedat}
  \qquad\qquad\quad
  \begin{alignedat}[t]{2}
    \text{\emph{Contexts}} &\quad&
      \ctx &\Coloneqq \ctxe* \mid \ctx , a{:}A
    \\
    \text{\emph{Addresses}} &&
      a , b , c , d &\Coloneqq x \mid \alpha
  \end{alignedat}
\end{equation*}

\paragraph*{Weakening and contraction.}
For reasons that will become clear in \cref{sec:SNAX}, we diverge slightly from the presentation of the semi-axiomatic sequent calculus and its SAX type theory seen in~\cite{DeYoung20fscd} and use explicit rules for weakening and contraction.
With respect to process syntax, both weakening and contraction are silent.
(The explicit rules for weakening and contraction mean that our SAX typing rules will not immediately yield a syntax-directed type checking algorithm.
But conversion to the implicit weakening and contraction of~\cite{DeYoung20fscd} is standard, and those rules are directly suitable for type checking.)
\begin{equation*}
  \begin{gathered}[t]
    \infer[\jrule{W}]{\ctx , A \vdash C}{
      \ctx \vdash C}
    \\[2ex]
    \infer[\jrule{C}]{\ctx , A \vdash C}{
      \ctx , A , A \vdash C}
  \end{gathered}
  \qquad\qquad\quad
  \begin{gathered}[t]
    \infer[\jrule{W}]{\ctx , a{:}A \vdash P :: (c : C)}{
      \ctx \vdash P :: (c : C)}
    \\[2ex]
    \infer[\jrule{C}]{\ctx , a{:}A \vdash P :: (c : C)}{
      \ctx , a{:}A , a{:}A \vdash P :: (c : C)}
  \end{gathered}
\end{equation*}

\paragraph*{Cut and identity.}
The semi-axiomatic sequent calculus's cut rule (inherited from the sequent calculus) corresponds to SAX's static typing rule for a notion of \emph{futures} for concurrent computation, $\cut{x}{P}{Q}$.
\begin{equation*}
  \infer[\jrule{CUT}]{\ctx_1 , \ctx_2 \vdash C}{
    \ctx_1 \vdash A & \ctx_2 , A \vdash C}
  \qquad\qquad\quad
  \infer[\jrule{CUT}]{\ctx_1 , \ctx_2 \vdash \cut{x}{P}{Q} :: (c : C)}{
    \ctx_1 \vdash P :: (x : A) & \ctx_2 , x{:}A \vdash Q :: (c : C) &
    \text{($x$ fresh)}}
\end{equation*}
Operationally, $\cut{x}{P}{Q}$ runs by first allocating memory for data of type $A$ and then running, in parallel, processes $P$ and $Q$ to write data to addresses $x$ and $c$, respectively.
If $Q$ tries to read from address $x$ before $P$ has written to $x$, then $Q$ will block until $P$ does write to $x$.\footnote{%
It is also possible to assign a sequential semantics (akin to call-by-value) or a call-by-need semantics to cuts~\cite{Pruiksma22jfp}.
In this paper, we will use the simplest and most general semantics, which is concurrent.}

The semi-axiomatic sequent calculus's identity rule (also inherited from the sequent calculus) corresponds to the SAX typing rule for a primitive operation, $\id{a}{b}$, for copying data from address $b$ to address $a$.
This copy operation is shallow: it does not follow pointers to copy recursively.
\begin{equation*}
  \infer[\jrule{ID}]{A \vdash A}{}
  \qquad\qquad\quad
  \infer[\jrule{ID}]{b{:}A \vdash \id{a}{b} :: (a : A)}{}
  \qquad
  \begin{tikzpicture}
    \node (a)
      [draw, label=left:$a$,
       rectangle split, rectangle split parts=1,
       rectangle split horizontal,
       rectangle split empty part width=0.5ex,
       rectangle split empty part height=1.5ex]
      {};
    \node (b)
      [draw, label=left:$b$,
       rectangle split, rectangle split parts=1,
       rectangle split horizontal,
       rectangle split empty part width=0.5ex,
       rectangle split empty part height=1.5ex,
       right=0.75cm of a]
      {};
    \node [at=(b)] {$\cellfiller$};
  \end{tikzpicture}
  \enspace\:\raisebox{1ex}{${}\rightsquigarrow{}$}\,
  \begin{tikzpicture}
    \node (a)
      [draw, label=left:$a$,
       rectangle split, rectangle split parts=1,
       rectangle split horizontal,
       rectangle split empty part width=0.5ex,
       rectangle split empty part height=1.5ex]
      {};
    \node [at=(a)] {$\cellfiller$};
    \node (b)
      [draw, label=left:$b$,
       rectangle split, rectangle split parts=1,
       rectangle split horizontal,
       rectangle split empty part width=0.5ex,
       rectangle split empty part height=1.5ex,
       right=0.75cm of a]
      {};
    \node [at=(b)] {$\cellfiller$};
  \end{tikzpicture}
\end{equation*}
As we have done here, we will use pictures throughout this section to provide intuition about the way that data is laid out in SAX;
in \cref{sec:SNAX}, these will serve as a point of comparison with the layouts offered by our SNAX type theory.

\paragraph*{Pairs, type $A_1 \tensor A_2$.}
In the Curry--Howard isomorphism between natural deduction and the simply typed $\lambda$-calculus, conjunctions $A_1 \land A_2$ are interpreted as product types $A_1 \tensor A_2$ that describe pairs of values of types $A_1$ and $A_2$, respectively, (as well as describing the expressions that evaluate to such values).
With the shift in perspective to the semi-axiomatic sequent calculus and shared memory concurrency, conjunction has a slightly different -- but very closely related -- interpretation: types $A_1 \tensor A_2$ describe addresses $a$ to which a pair $\tensorV[a_1,a_2]$ of addresses of types $A_1$ and $A_2$, respectively, must be written (as well as describing the processes that must write to such addresses $a$).

The $\arule{\land}$ axiom of the semi-axiomatic sequent calculus thus corresponds to a static typing rule for the process $\tensorA[a_1,a_2]{a}$ that writes a pair of addresses $\tensorV[a_1,a_2]$ to address $a$.
\begin{equation*}
  \infer[\arule{\land}]{A_1 , A_2 \vdash A_1 \land A_2}{}
  \qquad\quad\enspace
  \infer[\arule{\tensor}]{a_1{:}A_1 , a_2{:}A_2 \vdash \tensorA[a_1,a_2]{a} :: (a : A_1 \tensor A_2)}{}
  \quad\enspace
  \begin{tikzpicture}[baseline=(a.south)]
    \node (a)
      [draw, label=left:$a$,
       rectangle split, rectangle split parts=2,
       rectangle split horizontal,
       rectangle split empty part width=0.5ex,
       rectangle split empty part height=1.5ex]
      {\nodepart{two}};
  \end{tikzpicture}
  \enspace\raisebox{1ex}{${}\rightsquigarrow{}$}\,
  \begin{tikzpicture}[baseline=(a.south)]
    \node (a)
      [draw, label=left:$a$, 
       rectangle split, rectangle split parts=2,
       rectangle split horizontal,
       rectangle split empty part width=0.5ex,
       rectangle split empty part height=1.5ex]
      {\nodepart{two}};
    \node (a2) [overlay, above=0.4cm of a.two north, text depth=0pt] {$a_2$};
    \draw [overlay] (a.two north |- a.east) edge[->] (a2);
    \node (a1) [overlay, anchor=base, at=(a.mid |- a2.base), text depth=0pt] {$a_1$};
    \draw [overlay] (a.base |- a.east) edge[->] (a1);
  \end{tikzpicture}
\end{equation*}

As a proof-theoretic aside, in the semi-axiomatic sequent calculus, the above $\arule{\land}$ axiom is used in place of the sequent calculus's full-fledged right rule, $\rrule{\land}$, that has premises $\ctx \vdash A_1$ and $\ctx \vdash A_2$.
Using the cut and identity rules, the $\arule{\land}$ axiom and the usual right rule are interderivable and therefore characterize the same notion of intuitionistic conjunction.
The same pattern will hold for the other logical connectives in the semi-axiomatic sequent calculus: depending on whether the connective's polarity~\cite{Andreoliop92,Girard93} is positive or negative, either the sequent calculus right or left rule will be replaced with an equivalently expressive axiom.
Moreover, right axioms and right rules will write, whereas left axioms and left rules will read.

Under our shared memory interpretation, the left rule for conjunction becomes a static typing rule for the process $\tensorL[x_1,x_2]{a}{P}$ that reads the pair of addresses that is stored at address $a$:
\begin{equation*}
  \infer[\lrule{\land}]{\ctx , A_1 \land A_2 \vdash C}{
    \ctx , A_1 , A_2 \vdash C}
  \qquad\quad
  \infer[\lrule{\tensor}]{\ctx , a{:}A_1 \tensor A_2 \vdash \tensorL[x_1,x_2]{a}{P} :: (c : C)}{
    \ctx , x_1{:}A_1 , x_2{:}A_2 \vdash P :: (c : C)}
\end{equation*}
Operationally, the pair of addresses $\tensorV[a_1,a_2]$ stored at address $a$ is read from memory; then variables $x_1$ and $x_2$ are bound to addresses $a_1$ and $a_2$, respectively, and execution continues according to process $P$.

\vspace{\baselineskip}
\noindent\emph{Example.}
At this point, we can consider our first, very simple example process.
The commutativity of conjunction can be captured by a semi-axiomatic proof of $A_1 \land A_2 \vdash A_2 \land A_1$, and its computational content is a shared memory process of type $p : A_1 \tensor A_2 \vdash q : A_2 \tensor A_1$ that creates a new pair at address $q$ by swapping the components of the existing pair at address $p$:
\begin{equation*}
  \infer[\lrule{\land}]{A_1 \land A_2 \vdash A_2 \land A_1}{
    \infer[\arule{\land}]{A_1 , A_2 \vdash A_2 \land A_1}{}}
  \qquad\qquad\quad\!\!
  \begin{gathered}[b]
    p : A_1 \tensor A_2 \vdash
    \tensorL[x_1,x_2]{p}{
    \tensorA[x_2,x_1]{q}} :: (q : A_2 \tensor A_1)
  \\
  \begin{tikzpicture}
    \node (p)
      [draw, label={[text depth=0pt]left:$p$},
       rectangle split, rectangle split parts=2,
       rectangle split horizontal,
       rectangle split empty part width=0.5ex,
       rectangle split empty part height=1.5ex]
      {\nodepart{two}};
    \node (x2) [overlay, above=0.4cm of p.two north, text depth=0pt] {$x_2$};
    \draw [overlay] (p.two north |- p.east) edge[->] (x2);
    \node (x1) [overlay, anchor=base, at=(p.mid |- x2.base), text depth=0pt] {$x_1$};
    \draw [overlay] (p.base |- p.east) edge[->] (x1);
    \node (q)
      [draw, right=0.75cm of p, label={[text depth=0pt]left:$q$},
       rectangle split, rectangle split parts=2,
       rectangle split empty part width=0.5ex,
       rectangle split empty part height=1.5ex,
       rectangle split horizontal]
      {\nodepart{two}};
  \end{tikzpicture}
  \quad\raisebox{1ex}{${}\rightsquigarrow{}$}\enspace
  \begin{tikzpicture}
    \node (p)
      [draw, label={[text depth=0pt]left:$p$},
       rectangle split, rectangle split parts=2,
       rectangle split horizontal,
       rectangle split empty part width=0.5ex,
       rectangle split empty part height=1.5ex]
      {\nodepart{two}};
    \node (x2) [above=0.4cm of p.two north, text depth=0pt] {$x_2$};
    \draw (p.two north |- p.east) edge[->] (x2);
    \node (x1) [anchor=base, at=(p.mid |- x2.base), text depth=0pt] {$x_1$};
    \draw (p.base |- p.east) edge[->] (x1);
    \node (q)
      [draw, right=0.75cm of p, label={[text depth=0pt]left:$q$},
       rectangle split, rectangle split parts=2,
       rectangle split empty part width=0.5ex,
       rectangle split empty part height=1.5ex,
       rectangle split horizontal]
      {\nodepart{two}};
    \draw (q.two north |- q.east) edge[->] ([xshift=3pt]x1.south);
    \draw (q.base |- q.east) edge[->] ([xshift=3pt]x2.south);
  \end{tikzpicture}
  \end{gathered}
\end{equation*}
The process first reads from address $p$ and binds variables $x_1$ and $x_2$ to the pair of addresses that are stored there.
Then it writes those same addresses in reverse order as a pair at address $q$.

\paragraph*{Unit, type $\one$.}
The unit type $\one$ is the nullary form of the product type $A_1 \tensor A_2$ and arises in Curry--Howard correspondence with $\top$.
The $\arule{\top}$ axiom becomes a typing rule for the construct $\oneA{a}$, and the $\lrule{\top}$ rule (which, also being an instance of weakening, is uninteresting in terms of provability, but is computationally relevant) becomes a typing rule for the construct $\oneL{a}{P}$.
\begin{equation*}
  \begin{gathered}[t]
    \infer[\arule{\top}]{\ctxe \vdash \top}{}
    \\[2ex]
    \infer[\lrule{\top}]{\ctx , \top \vdash C}{
      \ctx \vdash C}
  \end{gathered}
  \qquad\qquad\quad
  \begin{gathered}[t]
    \infer[\arule{\one}]{\ctxe \vdash \oneA{a} :: (a : \one)}{}
    \\[2ex]
    \infer[\lrule{\one}]{\ctx , a{:}\one \vdash \oneL{a}{P} :: (c : C)}{
      \ctx \vdash P :: (c : C)}
  \end{gathered}
  \hspace{-1em}
  \begin{tikzpicture}
    \node (a)
      [draw, label=left:$a$,
       rectangle split, rectangle split parts=1,
       rectangle split horizontal,
       rectangle split empty part width=0.5ex,
       rectangle split empty part height=1.5ex]
      {\nodepart{two}};
  \end{tikzpicture}
  \enspace\raisebox{1ex}{${}\rightsquigarrow{}$}\,
  \begin{tikzpicture}
    \node (a)
      [draw, label=left:$a$,
       rectangle split, rectangle split parts=1,
       rectangle split horizontal,
       rectangle split empty part width=0.5ex,
       rectangle split empty part height=1.5ex]
      {\nodepart{two}};
    \node [overlay, anchor=base, at=(a.base), xshift=0.25pt, yshift=1pt] {$\oneV$};
  \end{tikzpicture}
\end{equation*}

\paragraph*{Tagged unions, type $\plus*[\ell \in L]{\ell\colon A_\ell}$.}
Disjunction corresponds to a labeled sum type, $\plus*[\ell \in L]{\ell\colon A_\ell}$, for tagged unions.
Being a positive connective, like conjunction and truth, disjunction's  $\arule{\lor\mkern-1mu}_k$ axiom in the semi-axiomatic sequent calculus becomes a typing rule for writing a tagged address.
Specifically, the process $\plusA[a_k]{a}{k}$ writes a tag $k$ and an address $a_k$ into memory at address $a$.
\begin{equation*}
  \infer[\arule{\lor}_k]{A_k \vdash A_1 \lor A_2}{
    (k \in \set{1,2})}
  \qquad\qquad\enspace
  \infer[\arule{\plus}_k]{a_k{:}A_k \vdash \plusA[a_k]{a}{k} :: (a : \plus*[\ell \in L]{\ell\colon A_\ell})}{
    (k \in L)}
  \qquad
  \begin{tikzpicture}
    \node (a)
      [draw, label=left:$a$,
       rectangle split, rectangle split parts=2,
       rectangle split empty part width=0.5ex,
       rectangle split empty part height=1.5ex,
       rectangle split horizontal]
      {\nodepart{two}};
  \end{tikzpicture}
  \enspace\:\raisebox{1ex}{${}\rightsquigarrow{}$}\,
  \begin{tikzpicture}
    \node (a)
      [draw, label=left:$a$,
       rectangle split, rectangle split parts=2,
       rectangle split empty part width=0.5ex,
       rectangle split empty part height=1.5ex,
       rectangle split horizontal]
      {\nodepart{two}};
    \node [anchor=base, at=(a.base)] {$k$};
    \node (ak) [above=0.4cm of a.two north, text depth=0pt, overlay] {$a_k$};
    \draw [overlay] (a.two north |- a.east) edge[->] (ak);
  \end{tikzpicture}
\end{equation*}
Symmetrically -- and adhering to the pattern for positive types -- the semi-axiomatic sequent calculus's $\lrule{\lor}$ rule becomes a typing rule for the reading construct $\plusL[\ell \in L]{a}{[x_\ell]{\ell} => P_\ell}$ that branches on the tag that it reads from address $a$.
\begin{equation*}
  \infer[\lrule{\lor}]{\ctx , A_1 \lor A_2 \vdash C}{
    \forallseq{\ell \in \set{1,2}}{\ctx , A_\ell \vdash C}}
  \qquad\qquad
  \infer[\lrule{\plus}]{\ctx , a{:}\plus*[\ell \in L]{\ell\colon A_\ell} \vdash \plusL[\ell \in L]{a}{[x_\ell]{\ell} => P_\ell} :: (c : C)}{
    \forallseq{\ell \in L}{\ctx , x_\ell{:}A_\ell \vdash P_\ell :: (c : C)}}
\end{equation*}
The process $\plusL[\ell \in L]{a}{[x_\ell]{\ell} => P_\ell}$ reads the tag, say $k \in L$, stored at address $a$ and selects the corresponding branch.
The variable $x_k$ is bound to the address that was tagged by $k$, and execution continues according to $P_k$.

\vspace{\baselineskip}
\noindent
\emph{Example.}
At this point, we can consider another simple example.
Booleans can be described with the type $\plus*{\mt{tt}\colon \one , \mt{ff}\colon \one}$, which we abbreviate as $\bool$.
The following process reads the boolean stored at address $a$ and then writes its negation at address $b$.
The diagram shows the process's execution when the tag stored at address $a$ is $\mt{tt}$; the other case is symmetric.
\begin{equation*}
  \begin{lgathered}[b]
    \mathrlap{\bool \defd \plus*{\mt{tt}\colon \one , \mt{ff}\colon \one}} \\
    \begin{array}[b]{@{}r@{}l@{}}
      a{:}\bool \vdash {} &
      \mathsf{read}\:a\:(\plusV[x]{\mt{tt}} \Rightarrow \plusA[x]{b}{\mt{ff}}\\[-1ex]
      & \hphantom{\mathsf{read}\:a\:(}\mathllap{{} \mid {}} \plusV[y]{\mt{ff}} \Rightarrow \plusA[y]{b}{\mt{tt}}) :: (b : \bool)
    \end{array}
  \end{lgathered}
  \qquad
  \begin{tikzpicture}
    \node (a)
      [draw, label=left:$a$,
       rectangle split, rectangle split parts=2,
       rectangle split horizontal,
       rectangle split empty part width=0.5ex,
       rectangle split empty part height=1.5ex]
      {$\mt{tt}$\nodepart{two}};
    \node (x) [right=0.4cm of a] {$x\mathrlap{{:}\one}$};
    \draw (a.two north |- a.east) edge[->] (x);

    \node (b)
      [draw, below right=0.25cm and 0cm of a.south west,
       label=left:$b$,
       rectangle split, rectangle split parts=2,
       rectangle split horizontal,
       rectangle split empty part width=0.5ex,
       rectangle split empty part height=1.5ex]
      {\nodepart{two}};
  \end{tikzpicture}
  \enspace\!\!\raisebox{1ex}{${}\rightsquigarrow{}$}\,
  \begin{tikzpicture}
    \node (a)
      [draw, label=left:$a$,
       rectangle split, rectangle split parts=2,
       rectangle split horizontal,
       rectangle split empty part width=0.5ex,
       rectangle split empty part height=1.5ex]
      {$\mt{tt}$\nodepart{two}};
    \node (x) [right=0.4cm of a] {$x{:}\one$};
    \draw (a.two north |- a.east) edge[->] (x);

    \node (b)
      [draw, below right=0.25cm and 0cm of a.south west,
       label=left:$b$,
       rectangle split, rectangle split parts=2,
       rectangle split horizontal,
       rectangle split empty part width=0.5ex,
       rectangle split empty part height=1.5ex]
      {$\mt{ff}$\nodepart{two}};

    \draw (b.two north |- b.east) edge[->, shorten >=-0.25ex] (x);
  \end{tikzpicture}
\end{equation*}

\paragraph*{Functions, type $A_1 \imp A_2$.}
Being a negative proposition, the implication $A_1 \limp A_2$ follows a story dual to that of the positive conjunction $A_1 \land A_2$.
Unlike the positive types, which write with axioms and read with left rules, the function type $A_1 \imp A_2$ that corresponds to implication writes with a right rule and reads with an axiom.

The semi-axiomatic sequent calculus's $\rrule{\limp}$ rule
therefore becomes a typing rule for the process $\impR{a}{x,z}{P}$.
In practice, this process might write (a pointer to) a closure to address $a$, but closures exist at a lower level of abstraction than the Curry--Howard correspondence between the semi-axiomatic sequent calculus and SAX supports.
For this reason, we think of the process $\impR{a}{x,z}{P}$ as writing the \emph{continuation} $\impK{x,z}{P}$.
\begin{equation*}
  \infer[\rrule{\limp}]{\ctx \vdash A_1 \limp A_2}{
    \ctx , A_1 \vdash A_2}
  \qquad\qquad\quad
  \infer[\rrule{\imp}]{\ctx \vdash \impR{a}{x,z}{P} :: (a : A_1 \imp A_2)}{
    \ctx , x{:}A_1 \vdash P :: (z : A_2)}
  \qquad
  \begin{tikzpicture}
    \node (a)
      [draw, label=left:$a$,
       rectangle split, rectangle split parts=1,
       rectangle split horizontal,
       rectangle split empty part width=0.5ex,
       rectangle split empty part height=1.5ex]
      {$\phantom{P\:}\vphantom{\rangle}$};
  \end{tikzpicture}
  \enspace\raisebox{1.25ex}{${}\rightsquigarrow{}$}\,
  \begin{tikzpicture}
    \node (a)
      [draw, label=left:$a$,
       rectangle split, rectangle split parts=1,
       rectangle split horizontal,
       rectangle split empty part width=0.5ex,
       rectangle split empty part height=1.5ex]
      {$\tensorV[x,z] \Rightarrow P$};
  \end{tikzpicture}
\end{equation*}

The semi-axiomatic sequent calculus's $\arule{\limp}$ axiom becomes a typing rule for the construct $\impA{a}{a_1,a_2}$.
\begin{equation*}
  \infer[\arule{\limp}]{A_1 \limp A_2 , A_1 \vdash A_2}{}
  \qquad\qquad\quad
  \infer[\arule{\imp}]{a{:}A_1 \imp A_2 , a_1{:}A_1 \vdash \impA{a}{a_1,a_2} :: (a_2 : A_2)}{}
\end{equation*}
This construct reads the continuation $\impK{x,z}{P}$ stored at address $a$ and passes it two addresses: $a_1$, the address of the function argument of type $A_1$ to which the continuation should be applied; and $a_2$, the address to which the called function should write its result of type $A_2$.
The variables $x$ and $z$ are bound to these addresses, respectively, and execution continues according to $P$.

\paragraph*{Other types.}
Other SAX types and process constructs also emerge from this Curry--Howard reading of the semi-axiomatic sequent calculus.
For example, it is possible to adapt the negative polarity conjunction from intuitionistic logic to a type of lazy records (e.g., like call-by-push-value~\cite{Levy01phd}).
Its SAX typing rules are dual to those for tagged unions~\cite{DeYoung20fscd}; moreover, because function types already exemplify the key aspects of negative types in SAX, we do not present the details of negative conjunction and lazy records in this paper.

Another possible SAX type is $\dn A$, which arises from the downshift of adjoint logic~\cite{Benton94csl,Reed09un,Pruiksma18aun}.
From a purely logical standpoint, $\dn A$ is not especially interesting because $\dn A$ is logically equivalent to $A$.
Neither is $\dn A$
particularly useful in SAX processes: it merely serves to introduce indirections beyond those already present in abundance in SAX.
However, the type $\dn A$ is also present in SNAX and becomes much more useful there, so we postpone further discussion of type $\dn A$, processes $\dnA{a}{b}$ and $\dnL{a}{x}{P}$, and their typing rules and operational semantics to \cref{sec:SNAX}.

\subsection{Adding recursion to SAX}

For most practical examples, recursively defined types and processes are needed.
Recursion in SAX goes beyond a strict Curry--Howard correspondence with the semi-axiomatic sequent calculus, but only in the same way that recursive functional programming goes beyond a strict Curry--Howard correspondence with natural deduction.

Instead of adding an explicit $\mu$, $\mathsf{fold}$, and $\mathsf{unfold}$ operators, we use recursive type definitions
and recursive process definitions.
Recursive type definitions have the form $t \defd A$; we choose to treat them equirecursively so that $t$ and its unfolding, $A$, are indistinguishable.
Under this interpretation, type definitions like $t = t$ would not be sensible, so SAX requires type definitions to be \emph{contractive}: each type name $t$ must (eventually) unfold to a logical type constructor like $\tensor$ or $\imp$.

Recursive process definitions have the form $\fun{\mathit{p}}{(z{:}C)}{(x_1{:}A_1) \dotsm (x_n{:}A_n)} = P$, where the argument $x_1{:}A_1, \dotsc, x_n{:}A_n$ may be read by process~$P$, and $z{:}C$ is the destination to which $P$ will write.
(If the process $P$ diverges, it escapes the obligation to write by postponing that obligation indefinitely.)
Calls to these recursively defined processes are made by the process construct $\callfun{\mathit{p}}{c}{a_1 \dotsb\mkern1.5mu a_n}$.
Its typing rule is the following where $\sig$ is a signature that holds recursive type and process definitions.
(See~\cref{sec:SAX:details} for more details on signatures.)
\begin{equation*}
  \infer[\jrule{CALL}]{a_1{:}A_1 , \dotsc , a_n{:}A_n \vdash_{\sig} \callfun{\mathit{p}}{c}{a_1 \dotsb\mkern1.5mu a_n} :: (c : C)}{
    (\fun{\mathit{p}}{(z{:}C)}{(x_1{:}A_1) \dotsm (x_n{:}A_n)} = P) \in \sig}
\end{equation*}
Operationally, $\callfun{\mathit{p}}{c}{a_1 \dotsb\mkern1.5mu a_n}$ will lookup the definition for $\mathit{p}$ and execution will continue according to the definition's body, substituting addresses for the argument and destination variables.
Accordingly, recursive process definitions are required to be contractive.

\subsection{Extended example: mapping a function across a linked list of booleans}

As an extended example of the SAX type theory, we can consider the recursive type $\boollist$ that describes linked lists of booleans.
(A polymorphic type of linked lists is not currently possible in SAX, but adding parametric polymorphism to SAX is a primary goal of future work.)
The type $\boollist$ is defined as follows.
\begin{equation*}
  \begin{lgathered}
  \boollist \defd \plus*{\mt{nil}\colon \one ,\, \mt{cons}\colon \bool \tensor \boollist}
  \\
  \text{\emph{(a) }}
  \begin{tikzpicture}[baseline=(xs.base)]
    \node (xs)
      [draw, label=left:$\mathit{xs}$,
       rectangle split, rectangle split parts=2,
       rectangle split empty part width=0.5ex,
       rectangle split empty part height=1.5ex,
       rectangle split horizontal]
      {$\mt{nil}$\nodepart{two}};

    \node (u) [right=0.4cm of xs] {$u\smash{{:}\one}$};
    \draw (xs.two north |- xs.east) edge[->] (u);
  \end{tikzpicture}
  \qquad
  \text{\emph{(b) }}
  \begin{tikzpicture}[baseline=(xs.base)]
    \node (xs)
      [draw, label=left:$\mathit{xs}$,
       rectangle split, rectangle split parts=2,
       rectangle split empty part width=0.5ex,
       rectangle split empty part height=1.5ex,
       rectangle split horizontal]
      {$\mt{cons}$\nodepart{two}};

    \node (p) [right=0.4cm of xs, text depth=0pt] {$p$};
    \node (pn)
      [right=0.02cm of p.east,
       draw,
       rectangle split, rectangle split parts=2,
       rectangle split empty part width=0.5ex,
       rectangle split empty part height=1.5ex,
       rectangle split horizontal]
      {\nodepart{two}};
    \draw (xs.two north |- xs.east) edge[->] (p);

    \node (x) [overlay, above=0.4cm of {pn.mid |- pn.north}, text depth=0pt] {$x\mathrlap{{:}\bool}$};
    \draw (pn.base |- pn.east) edge[->] (x);

    \node (xs') [right=0.4cm of pn] {$\mathit{xs}\smash{'{:}\boollist}$};
    \draw (pn.two north |- pn.east) edge[->] (xs');
  \end{tikzpicture}
  \end{lgathered}
\end{equation*}
Specifically, an address $\mathit{xs}$ of type $\boollist$ stores either:
\begin{enumerate*}[label=\emph{(\alph*)}]
\item the tag $\mt{nil}$ and an address $u$ of type $\one$; or
\item the tag $\mt{cons}$ and an address $p$ of type $\bool \tensor \boollist$, which itself stores a pair of addresses of types $\bool$ and $\boollist$, respectively.
\end{enumerate*}
Aside from the high degree of indirection and the use of tags to replace null pointers, this is a fairly recognizable representation of a linked list of booleans.

A $\mathit{map}$ function for mapping a unary boolean function $f$ across a linked list $\mathit{xs}$ of booleans and writing the resulting list to destination $\mathit{ys}$ is given by the following recursive definition.
\begin{equation*}
    \begin{array}[b]{@{}l@{\hspace{2em}}l@{}}
      \multicolumn{2}{l}{\fun{\mathit{map}}{(\mathit{ys} : \boollist)}{(f : \bool \imp \bool)\:(\mathit{xs} : \boollist)} = {}}\\[-1ex]
      \quad\mathsf{read}\:\mathit{xs}\:(                        & \mbox{\itshape\% read and branch on tag at $\mathit{xs}$}\\[-1ex]
      \quad\hphantom{\mid {}} \plusV[u]{\mt{nil}} \Rightarrow   
                   \id{\mathit{ys}}{\mathit{xs}}                & \mbox{\itshape\% copy (empty) input list $\mathit{xs}$ to destination $\mathit{ys}$} \\[-1ex]
      \quad\mid \plusV[p]{\mt{cons}} \Rightarrow
            \tensorL[x,\mathit{xs}']{p}{                        & \mbox{\itshape\% read pair at $p$} \\[-1ex]
      \hphantom{\quad\mid \plusV[p]{\mt{cons}} \Rightarrow {}}
            \cut{y}{\impA{f}{x,y}}{                             & \mbox{\itshape\% allocate $y$ and call $f$ with destination $y$}\\[-1ex]
      \hphantom{\quad\mid \plusV[p]{\mt{cons}} \Rightarrow {}}
            \cut{\mathit{ys}'}{
              \callfun{\mathit{map}}{\mathit{ys}'}{f\:\mathit{xs}'}}{ & \mbox{\itshape\% allocate $\mathit{ys}'$ and call $\mathit{map}$ recursively with dest.\ $\mathit{ys}'$}\\[-1ex]
      \hphantom{\quad\mid \plusV[p]{\mt{cons}} \Rightarrow {}}
            \cut{q}{\tensorA[y,\mathit{ys}']{q}}{              & \mbox{\itshape\% allocate $q$ and write pair of $y$ and $\mathit{ys}'$}\\[-1ex]
      \hphantom{\quad\mid \plusV[p]{\mt{cons}} \Rightarrow {}}
            \plusA[q]{\mathit{ys}}{\mt{cons}}}}}})             & \mbox{\itshape\% write the tagged pair to the original destination $\mathit{ys}$}
    \end{array}
\end{equation*}

\subsection{Details of the SAX type theory}\label{sec:SAX:details}

The types in SAX are as described above; 
recursive type definitions $t \defd A$ and recursive process definitions $\fun{\mathit{p}}{z}{x_1 \dotsb\mkern1.5mu x_n} = P$ are collected in signatures $\sig$.
A signature indexes the SAX typing judgment: $\ctx \vdash_{\sig} P :: (a : A)$.
However, because none of the typing rules affect the signature, it is frequently elided.
\begin{alignat*}{2}
  \text{\emph{Types}} &\quad&
    A , B , C &\Coloneqq A \tensor B \mid \one \mid \plus*[\ell \in L]{\ell\colon A_\ell} \mid \dn A
                    \mid A \imp B
                    \mid t
  \\
  \text{\emph{Signatures}} &&
    \sig &\Coloneqq \sige* \mid \sig , t \mathbin{\defd} A \mid \fun{\mathit{p}}{z}{x_1 \dotsb\mkern1.5mu x_n} \mathbin{=} P
\end{alignat*}

SAX processes $P$ and $Q$ have one of five forms: allocations, copies, writes, reads, and calls.
\begin{alignat*}{2}
  \text{\emph{Processes}} &\quad&
    P , Q  &\Coloneqq \cut{x}{P}{Q} \mid \id{a}{b}
                 \mid \writeV{a}{S} \mid \readV{a}{T}
                 \mid \callfun{\mathit{p}}{d}{a_1 \dotsb\mkern1.5mu a_n}
  \\
  \text{\emph{Storables}} &&
    S &\Coloneqq \tensorV[a_1, a_2] \mid \oneV \mid \plusV[a]{k} \mid \dnV{a}
            \mid \impK{x,z}{P}
  \\
  \text{\emph{Co-storables}} &&
    T &\Coloneqq \tensorK[x_1, x_2]{P} \mid \oneK{P} \mid \plusK[\ell \in L]{[x_\ell]{\ell} => P_\ell} \mid \dnK{x}{P}
            \mid \impV{a_1,a_2}
\end{alignat*}
Writes rely on a syntactic category of storables $S$, which are the data that may be written into a memory cell.
There is one storable for each type constructor: pairs of addresses, $\tensorV[a_1,a_2]$; unit value, $\oneV$; tagged address, $\plusV[a]{k}$; pointers, $\dnV{a}$; and function continuations, $\impK{x,z}{P}$.
Reads rely on a syntactic category of co-storables $T$ that are dual to storables.
Once again, there is one form of co-storable for each type constructor: continuations for pairs, unit values, tagged unions, and pointers; and pairs of addresses to be passed to function continuations.
We will not repeat the process typing rules here.

The operational semantics of the SAX type theory is based on multiset rewriting, using three semantic objects:
$\thread{a}{P}$ denotes a running process $P$ that must write to address $a$;
$\cell{a}{}$ denotes an empty memory cell at address $a$, \ie, one that has been allocated but not yet written;
and $\cell{a}{S}$ denotes a filled memory cell at address $a$, \ie, one that has been allocated and now stores $S$.
The `$\bang$' is notation borrowed from linear logic and denotes that filled cells implicitly persist across rewriting steps.

A configuration $\cnf$ is a collection of filled cells and threads with corresponding empty cells:
\begin{alignat*}{2}
  \text{\emph{Configurations}} &\quad&
    \cnf &\Coloneqq \cnfe* \mid \cnf_1 \cc \cnf_2
               \mid \thread{a}{P} \cc \cell{a}{}
               \mid \cell{a}{S}
  \\
  \text{\emph{Configuration contexts}} &&
    \cctx &\Coloneqq \cctxe* \mid \cctx , a{:}A
\end{alignat*}
We say that a configuration $\cnf$ is \emph{final} if it consists only of filled cells $\cell{a}{S}$.
Configurations are typed with a judgment $\cctx \vDash \cnf :: \cctx'$.
Configuration contexts $\cctx$ have the same syntactic structure as contexts $\ctx$, but configuration contexts $\cctx$ are not subject to contraction; in the judgment $\cctx \vDash \cnf :: \cctx'$, the addresses in $\cctx$ are therefore presumed to be distinct.
Also, in $\cnf$ and $\cctx$, we write $a$, $b$, and $c$ for readability, but they are all fresh runtime addresses $\alpha$ (not variables $x$).
The configuration typing judgment has the rules:
\begin{mathpar}
  \infer[\jrule{EMP}]{\cctx \vDash \cnfe* :: \cctx}{}
  \and
  \infer[\jrule{JOIN}]{\cctx \vDash \cnf_1 \cc \cnf_2 :: \cctx''}{
    \cctx \vDash \cnf_1 :: \cctx' & 
    \cctx' \vDash \cnf_2 :: \cctx''}
  \\
  \infer[\jrule{THREAD}]{\cctx \vDash \thread{a}{P} \cc \cell{a}{} :: (\cctx , a{:}A)}{
    (a \notin \dom{\cctx}) & (\cctx \supseteq \ctx) & \ctx \vdash P :: (a : A)}
  \and
  \infer[\jrule{CELL}]{\cctx \vDash \cell{a}{S} :: (\cctx , a{:}A)}{
    (a \notin \dom{\cctx}) & (\cctx \supseteq \ctx) & \cctx \vdash \writeV{a}{S} :: (a : A)}
\end{mathpar}
Notice that both the $\jrule{THREAD}$ and $\jrule{CELL}$ rules check that the address $a$ is not already present in the domain of $\ctx$, to ensure that each address has a unique type.
Both rules also rely on the static typing judgment for SAX processes.

The concurrent operational semantics is then described by the following multiset rewriting clauses.
\begin{equation*}
  \begin{lgathered}[t]
    \begin{aligned}[t]
      \MoveEqLeft[.75]
      \thread{c}{(\cut{x}{P}{Q})} \qquad\qquad\qquad\qquad\quad\;\,\text{($\alpha$ fresh)} \\
        &\stepsto \thread[\alpha/x]{\alpha}{P} \cc \cell{\alpha}{} \cc \thread[\alpha/x]{c}{Q}
    \end{aligned}
    \\[0.5ex]
    \thread{a}{\id{a}{b}} \cc \cell{a}{} \cc \cell{b}{S}
      \stepsto \cell{a}{S}
    \\[0.5ex]
    \thread{a}{\writeV{a}{S}} \cc \cell{a}{}
      \stepsto \cell{a}{S}
    \\[0.5ex]
    \thread{c}{\readV{a}{T}} \cc \cell{a}{S}
      \stepsto \thread{c}{\pass{S}{T}}
    \\[0.5ex]
    \begin{aligned}[t]
      \MoveEqLeft[.75]
      \thread{c}{\callfun{\mathit{p}}{c}{a_1 \dotsb\mkern1.5mu a_n}} \\
        &\stepsto \thread{c}{[c/z, a_1/x_1, \dotsc, a_n/x_n]P} \\
        &\text{(where $\fun{\mathit{p}}{z}{x_1 \dotsb\mkern1.5mu x_n} = P$)}
    \end{aligned}
  \end{lgathered}
  \qquad
  \begin{lgathered}[t]
    \text{where $\pass{S}{T}$ is given by } \\[0.5ex]\quad
    \begin{alignedat}[t]{2}
      \pass{\tensorV[a_1,a_2]}{\tensorK[x_1,x_2]{P}} &= [a_1/x_1,a_2/x_2]P \\[0.5ex]
      \pass{\oneV}{\oneK{P}} &= P \\[0.5ex]
      \pass{\plusV[a]{k}}{\plusK[\ell \in L]{[x_\ell]{\ell} => P_\ell}} &= [a/x_k]P_k \quad(k \in L) \\[0.5ex]
      \pass{\ptrV{a}}{\ptrK{x}{P}} &= [a/x]P \\[0.5ex]
      \pass{\impK{x,z}{P}}{\impV{a_1,a_2}} &= [a_1/x,a_2/z]P
    \end{alignedat}
  \end{lgathered}
\end{equation*}
Preservation and progress hold for the SAX type theory~\cite{DeYoung20fscd}.
\begin{theorem}[Preservation]
  If $\cctx_0 \vDash \cnf :: \cctx$ and $\cnf \stepsto \cnf'$, then $\cctx_0 \vDash \cnf' :: \cctx'$ for some $\cctx' \supseteq \cctx$.
\end{theorem}

\begin{theorem}[Progress]
  If\/ $\vDash \cnf :: \cctx$, then either $\cnf$ is final or $\cnf \stepsto \cnf'$ for some $\cnf'$.
\end{theorem}

\section{SNAX type theory for data layout}\label{sec:SNAX}

The SAX type theory does not take layout considerations into account in the sense that addresses
remain entirely abstract.
As a concrete example, a single $\mt{cons}$ node in a linked list of booleans
can be thought of as laid out in SAX with three indirections, while a more compact flat layout
would be more memory-efficient:
\begin{equation*}
  \text{\emph{Instead of}} \enspace
  \begin{tikzpicture}[baseline=(c.base)]
    \node (c)
      [draw,
       rectangle split, rectangle split parts=2,
       rectangle split horizontal,
       rectangle split empty part width=0.5ex,
       rectangle split empty part height=1.5ex]
      {$\mt{cons}$\nodepart{two}};

    \node (p)
      [draw, right=0.25cm of c,
       rectangle split, rectangle split parts=2,
       rectangle split horizontal,
       rectangle split empty part width=0.5ex,
       rectangle split empty part height=1.5ex]
      {\nodepart{two}};

    \node (xs')
      [right=0.25cm of p,
       rectangle split, rectangle split parts=1,
       rectangle split empty part width=0.5ex,
       rectangle split empty part height=1.5ex,
       rectangle split horizontal]
      {${\dotso}$};

    \node (b)
      [draw, above right=0.25cm and 0cm of p.north west,
       rectangle split, rectangle split parts=2,
       rectangle split empty part width=0.5ex,
       rectangle split empty part height=1.5ex,
       rectangle split horizontal]
      {$\mt{tt}$\nodepart{two}};

    \node (u)
      [draw, right=0.25cm of b,
       rectangle split, rectangle split parts=1,
       rectangle split empty part width=0.5ex,
       rectangle split empty part height=1.5ex,
       rectangle split horizontal]
      {};
    \node [overlay, at=(u)] {$\oneV$};

    \draw (c.two south |- c.east) edge[->] (p);
    \draw (p.mid |- p.east) edge[->] (p.mid |- b.south);
    \draw (b.two south |- b.east) edge[->] (u);
    \draw (p.two south |- p.east) edge[->] (xs');
  \end{tikzpicture}
  \text{\emph{, the flat layout }}\,
  \begin{tikzpicture}[baseline=(c.base)]
    \node (c)
      [draw,
       rectangle split, rectangle split parts=3,
       rectangle split horizontal,
       rectangle split empty part width=0.5ex,
       rectangle split empty part height=1.5ex]
      {$\mt{cons}$\nodepart{two}$\smash{\mt{tt}}\vphantom{\mt{cons}}$\nodepart{three}};

    \node (xs')
      [right=0.25cm of c,
       rectangle split, rectangle split parts=1,
       rectangle split empty part width=0.5ex,
       rectangle split empty part height=1.5ex,
       rectangle split horizontal]
      {${\dotso}$};

    \draw (c.three south |- c.east) edge[->] (xs');
  \end{tikzpicture}
  \text{\emph{ is likely preferable.}}
\end{equation*}
SAX's rather extreme level of indirection arises from its heavy reliance on the $\jrule{CUT}$ rule.
This suggests that if we want to account for data layout in a SAX-like type theory, an understanding of cut elimination and the structure of cut-free proofs in the semi-axiomatic sequent calculus may provide some insight.

\subsection{Cut elimination in the semi-axiomatic sequent calculus}

Unfortunately, the semi-axiomatic sequent calculus does not immediately satisfy Gentzen-style cut elimination. 
As a counterexample, there is a semi-axiomatic proof of $B , B \limp A_1 , A_2 \vdash A_1 \land A_2$, namely
\begin{equation*}
  \infer[\jrule{CUT}]{B , B \limp A_1 , A_2 \vdash A_1 \land A_2}{
    \infer[\arule{\limp}]{B , B \limp A_1 \vdash A_1}{} &
    \infer[\arule{\land}]{A_1 , A_2 \vdash A_1 \land A_2}{}}
  \,,
\end{equation*}
but there is no \emph{cut-free} proof: the cut that appears here is essential and cannot be eliminated.

However, notice that the above cut has a \emph{subformula property}: the cut formula (here $A_1$) is a proper subformula of one of the conclusion sequent's formulas (here $A_1 \land A_2$).
Moreover, this subformula property derives from the use of the cut formula within the occurrence of the $\arule{\land}$ axiom.
Such formulas that are used by axioms are said to be \emph{eligible} to act as the cut formula in one of these well-behaved cuts, which are called \emph{snips}.
Proof-theoretically, both eligibility and snips are thought of as \emph{properties} of proofs: we can implicitly ignore that a formula is eligible and that a cut is a snip because eligibility is not part of the structure of the proof rules themselves.

Although semi-axiomatic proofs cannot be transformed to be fully cut-free, they can be transformed so that the only remaining cuts are these well-behaved snips.
For example, the above proof can be reformulated using a snip as follows; the eligible formulas are indicated by underlining.
\begin{equation*}
  \infer[\jrule{SNIP}]{\eli{B} , B \limp A_1 , \eli{A_2} \vdash A_1 \land A_2}{
    \infer[\arule{\limp}]{\eli{B} , B \limp A_1 \vdash \eli{A_1}}{} &
    \infer[\arule{\land}]{\eli{A_1} , \eli{A_2} \vdash A_1 \land A_2}{}}
\end{equation*} 
For more on the proof-theoretic view of eligibility as a property of proofs, we refer the reader to~\cite{DeYoung20fscd}.

\subsection{Making eligibility first-class}

As we saw in \cref{sec:SAX}, the SAX type theory is a Curry--Howard interpretation of the semi-axiomatic sequent calculus that ignores eligibility and snips, viewing them as mere technical refinements used in recovering (a modified form of) cut elimination.
However, in this section, we elevate eligibility and snips to be first-class concepts in the proof theory and thereby obtain, by Curry--Howard correspondence, a type theory that gives a clean, logically grounded account of data layout in shared memory concurrency.
We will call the resulting type theory SNAX, short for ``SAX with snips''.

\paragraph*{Eligibility, addresses, contexts, and judgments.}
To make eligibility a structural component of the semi-axiomatic sequent calculus, contexts $\ctx$ now contain ordinary antecedents $A$ and eligible antecedents $\eli{A}$ -- eligibility is no longer a refinement property, but instead intrinsic to an antecedent.
Owing to the subformula structure that underlies eligibility, an eligible antecedent $\eli{A}$ in a proof will correspond to an eligible address  $\eli{a{\cdot}p:A}$, where $p$ is a \emph{projection} and $a \cdot p$ is a new form of address.
For example, we will see shortly that just as $A_1$ is an eligible antecedent within the axiom for $A_1 \land A_2$, so will $a \cdot \pi_1$ be the address of the first component of a pair that itself begins at address $a$.
(We will sometimes elide the $\cdot$ operator in $a \cdot p$, especially when $\pi_1$ and $\pi_2$ are involved.)
\begin{equation*}
  \begin{alignedat}[t]{2}
    \text{\emph{Contexts}} &\quad&
      \ctx &\Coloneqq \ctxe* \mid \ctx , A \mid \ctx , \eli{A}
  \end{alignedat}
  \qquad\qquad
  \begin{alignedat}[t]{2}
    \text{\emph{Contexts}} &\quad&
      \ctx &\Coloneqq \ctxe* \mid \ctx , a{:}A \mid \ctx , \smash{\eli{a{\cdot}p{:}A}}\vphantom{p}
    \\
    \text{\emph{Addresses}} &&
      a , b , c , d &\Coloneqq x \mid \alpha \mid a \cdot p
    \\
    \text{\emph{Projections}} &&
      p &\Coloneqq \pi_1 \mid \pi_2 \mid \smash{\ptagg{\ell}} \vphantom{\ell} \mid p_1 \cdot p_2
  \end{alignedat}
\end{equation*}
We say that \emph{$a$ strictly extends $c$} and write $a \extends c$ whenever $a = c \cdot p$ for some projection $p$;
we also say that \emph{$a$ (weakly) extends $c$} and write $a \wextends c$ whenever either $a = c$ or $a \extends c$.
Note that every address has a variable (either a static $x$ or a runtime $\alpha$) at its head.
This reflects the idea that memory will be allocated in blocks and, at the level of abstraction that SNAX provides, each address exists relative to the base address of the block to which it belongs.

Projections provide a clean, logical description of the \emph{abstract} layout of data.
On the other hand, machine code relies on \emph{concrete} address arithmetic.
To bridge this gap, we envision that, at a lower level of abstraction, projections would be concretized.
One simple concretization $(a)^\star$ underlying our diagrams would be the following:
\begin{equation*}
  \begin{alignedat}[t]{2}
    (a \cdot \smash{\ptagg{k}}\vphantom{k})^\star &= a^\star + 1 \\
    (a \cdot \pi_1)^\star &= a^\star \\
    (a \cdot \pi_2)^\star &= a^\star + \size{A_1} &\enspace& \text{(when $a : A_1 \tensor A_2$)}
  \end{alignedat}
  \quad\!\text{where}\quad
  \begin{alignedat}[t]{3}
    \size{A_1 \tensor A_2} &= \size{A_1} + \size{A_2} &\quad&&
    \size{\dn A} &= 1 \\
    \size{\one} &= 0 &&&
    \size{A_1 \imp A_2} &= 1 \\
    \size{\plus*[\ell \in L]{\ell\colon A_\ell}} &= 1 + {\textstyle \max_{\ell \in L} \,\size{A_\ell}}
  \end{alignedat}
\end{equation*}
However, this is not the only possible concretization.
Practical compilers employ layout optimizations, such as bit-packing to reduce the space needed for the tags of consecutive tagged unions.
By using projections, SNAX is agnostic about the particular concretization chosen and remains flexible.

Aside from the changes to addresses and contexts, sequents and typing judgments look the same as in SAX: $\ctx \vdash A$ and $\ctx \vdash P :: (a : A)$, respectively.
As a general convention, for each typing judgment $\ctx \vdash P :: (a : A)$, we presuppose that $a \nwextends b$ for all $b{:}B \in \ctx$; the typing rules will maintain this invariant.
(This property is also provable for all $\eli{b{:}B} \in \ctx$ [see Lemma~\ref{lem:extends}].)
If we did not have this well-formedness condition, there would incorrectly be two writers to address $a$: one implicitly as part of the writer to $b$, and the other being $P$.

To achieve (its modified form of) cut elimination, the semi-axiomatic sequent calculus also uses eligible succedents.
However, because the SNAX type theory does not model the internal layout of function closures, eligible succedents do not appear in SNAX.
As a consequence, cut elimination (even in modified form) will not hold for SNAX.
Since we are primarily interested in operational aspects, we avoid here the technical complications
required to restore it, since type preservation and progress \emph{do} hold.

\paragraph*{Weakening and contraction.}

To correctly maintain the connection between an eligible antecedent and the axiom from which it derives its eligibility, some care must be taken with weakening and contraction.
Ordinary antecedents and their corresponding addresses are still subject to weakening, but eligible ones may \emph{not} be weakened away (or otherwise, their eligibility would fail to be derived from an axiom).
\begin{gather*}
  \infer[\jrule{W}]{\ctx , A \vdash C}{
    \ctx \vdash C}
  \qquad\qquad\quad
  \infer[\mathrlap{\jrule{W}}]{\ctx , a{:}A \vdash P :: (c : C)}{
    \ctx \vdash P :: (c : C)}
  \\[2ex]
  \text{(no weakening for $\eli{A}$ and $\eli{a{:}A}$)}
\end{gather*}
Contraction for two ordinary antecedents and their corresponding addresses occurs as in SAX.
Furthermore, contracting an eligible antecedent $\eli{A}$ and an ordinary antecedent $A$ together into $\eli{A}$ is permitted: the eligibility of the resulting $\eli{A}$ will still be traceable to an axiom.
But contracting two copies of $\eli{a{:}A}$ into one is not permitted in SNAX.
Such a contraction rule would not be harmful, but neither would it be useful
since no two eligible antecedents can be assigned the same address in a derivable judgment (see Lemma~\ref{lem:twoeli}).
In this way, eligible antecedents $\eli{A}$ and their corresponding addresses $\eli{a{:}A}$ follow an almost linear discipline.
\begin{gather*}
  \infer[\jrule{C}]{\ctx , A \vdash C}{
    \ctx , A , A \vdash C}
  \qquad\qquad\quad
  \infer[\jrule{C}]{\ctx , a{:}A \vdash P :: (c : C)}{
    \ctx , a{:}A , a{:}A \vdash P :: (c : C)}
  \\[2ex]
  \infer[\jrule{C}_{\jrule{E}}]{\ctx , \eli{A} \vdash C}{
    \ctx , \eli{A} , A \vdash C}
  \qquad\qquad\quad
  \infer[\jrule{C}_{\jrule{E}}]{\ctx , \eli{a{:}A} \vdash P :: (c : C)}{
    \ctx , \eli{a{:}A} , a{:}A \vdash P :: (c : C)}
  \\[2ex]
  \text{(no contraction for $\eli{A}, \eli{A}$ and $\eli{a{:}A} , \eli{a{:}A}$)}
\end{gather*}

\paragraph*{Cut, snip, and identity.}
The essential idea of the SNAX type theory is that, once eligibility and snips are first-class, we have the freedom to give different computational interpretations to cuts and snips.
SNAX retains the $\jrule{CUT}$ rule from SAX but also has a distinct $\p{\jrule{SNIP}}$ rule.
The $\jrule{CUT}$ rule and $\cut{x}{P}{Q}$ construct continue to act as in SAX, allocating memory for data of type $A$ and then running $P$ and $Q$ in parallel.
(The type $A$ should be inferred or given, if SNAX is used as a source language, so that allocation can depend on the type.)
\begin{equation*}
  \infer[\jrule{CUT}]{\ctx_1 , \ctx_2 \vdash C}{
    \ctx_1 \vdash A & \ctx_2 , A \vdash C}
  \qquad\qquad\quad
  \infer[\jrule{CUT}]{\ctx_1 , \ctx_2 \vdash \cut{x}{P}{Q} :: (c : C)}{
    \ctx_1 \vdash P :: (x : A) & \ctx_2 , x{:}A \vdash Q :: (c : C) &
    \text{($x$ fresh)}}
\end{equation*}
But the $\p{\jrule{SNIP}}$ rule is different.
Its $\snip{a}{P}{Q}$ construct does not (re-)allocate memory at address $a$ because it is an eligible address and, as such, refers to a location \emph{within} a block already allocated by an earlier $\jrule{CUT}$ rule.
Instead, it simply runs processes $P$ and $Q$ in parallel, with reads at address $a$ that occur in $Q$ blocking until that address has been written to by $P$.%
\footnote{Similar to the situation for SAX, it is possible to give sequential and call-by-need semantics to the SNAX $\jrule{CUT}$ and $\p{\jrule{SNIP}}$ constructs.}
\begin{equation*}
  \infer[\p{\jrule{SNIP}}]{\ctx_1 , \ctx_2 \vdash C}{
    \ctx_1 \vdash A & \ctx_2 , \eli{A} \vdash C}
  \qquad\qquad\quad
  \infer[\p{\jrule{SNIP}}]{\ctx_1 , \ctx_2 \vdash \snip{a}{P}{Q} :: (c : C)}{
    \ctx_1 \vdash P :: (a : A) & \ctx_2 , \eli{a{:}A} \vdash Q :: (c : C)}
\end{equation*}
In the pure proof theory, there would also be a symmetric $\n{\jrule{SNIP}}$ rule for eligible succedents.
However, as previously mentioned, SNAX ignores eligibility in succedents and therefore does not include $\n{\jrule{SNIP}}$.

Eligibility does not enter into the identity rule; it remains as in SAX.
(Note that the $\jrule{ID}$ typing rule is one place where it is essential to have the presupposition on typing judgments $\ctx \vdash P :: (a : A)$ that $a \nwextends b$ for all $b{:}B \in \ctx$.
Without it, the $\jrule{ID}$ rule would not be operationally sensible.)
\begin{equation*}
  \infer[\jrule{ID}]{A \vdash A}{}
  \qquad\qquad\qquad
  \infer[\jrule{ID}]{b{:}A \vdash \id{a}{b} :: (a : A)}{}
\end{equation*}

\paragraph*{Pairs, type $A_1 \tensor A_2$.}
Recall from \cref{sec:SAX} that SAX uses the construct $\tensorA[a_1,a_2]{a}$ to write a pair of addresses $\tensorV[a_1,a_2]$ to address $a$.
Because these addresses point to the data that are conceptually the components of a pair of values, the layout is very indirect.

In contrast, for SNAX, we first observe that the purely logical axiom $\arule{\land}$ has both antecedents $A_1$ and $A_2$ eligible (as denoted by the underlining), because both are proper subformulas that are used in an axiom.
This subformula structure is then reflected in the $\arule{\tensor}$ typing rule by assigning addresses $a \pi_1$ and $a \pi_2$ to $A_1$ and $A_2$, respectively.
Because these addresses are locally calculable from $a$, they are not needed in the process syntax $\tensorA{a}$.
\begin{equation*}
  \infer[\arule{\land}]{\eli{A_1} , \eli{A_2} \vdash A_1 \land A_2}{}
  \qquad\quad\;\;
  \infer[\arule{\tensor}]{\eli{a \pi_1{:}A_1} , \eli{a \pi_2{:}A_2} \vdash \tensorA{a} :: (a : A_1 \tensor A_2)}{}
  \qquad\!\!\!\!
  \begin{tikzpicture}[baseline=(a.south)]
    \node (a)
      [draw, label=left:$a$,
       rectangle split, rectangle split parts=6,
       rectangle split empty part width=0.5ex,
       rectangle split empty part height=1.5ex,
       rectangle split horizontal]
      {\nodepart{two}${\dotso}$\nodepart{three}\nodepart{four}\nodepart{five}${\dotso}$\nodepart{six}};
    \node (ap1) [overlay, below right=2pt of a.south west, inner xsep=0pt]
      {$a\pi_1$};
    \node (ap2) [overlay, below right=2pt of a.three split south, inner xsep=0pt]
      {$a\pi_2$};
    \node [overlay, above right=0pt of a.north west, inner xsep=0pt]
      {$\overbrace{\qquad\qquad\enspace}^{\smash{\text{\pbox[b]{4em}{data of\\[-1.5ex]type $A_1$}}}}$};
    \node [overlay, above left=0pt of a.north east, inner xsep=0pt]
      {$\overbrace{\qquad\qquad\enspace}^{\smash{\text{\pbox[b]{4em}{data of\\[-1.5ex]type $A_2$}}}}$};
    \path (ap1.mid) -- (ap2.mid);
    \path ([yshift=1.5ex]a.north west) -- ([yshift=1.5ex]a.north east);
  \end{tikzpicture}
\end{equation*}

The above picture depicts a particular flat layout for SNAX pairs in which $a \pi_1$ precedes $a \pi_2$.
At a formal level, SNAX abstracts away from these particulars: SNAX requires only that $a \pi_1$ and $a \pi_2$ are calculable from $a : A_1 \tensor A_2$ and that $a \pi_1 \cdot p_1 \neq a \pi_2 \cdot p_2$ for all projections $p_1$ and $p_2$.
For example, as far as SNAX is concerned, an equally correct picture could have $a \pi_2$ preceding $a \pi_1$.

SNAX also tweaks the construct for reading an address $a$ of type $A_1 \tensor A_2$.
Instead of binding variables with $\tensorL[x_1,x_2]{a}{P}$, we now use $\tensorL{a}{P}$.
Bound variables are no longer needed because we know that a SNAX pair at address $a$ will always have its components located at the relative addresses $a \pi_1$ and $a \pi_2$.
Logically, though, the static typing rule for reading at type $A_1 \tensor A_2$ is still derived from the semi-axiomatic sequent calculus's $\lrule{\land}$ rule.
\begin{equation*}
  \infer[\lrule{\land}]{\ctx , \deli{A_1 \land A_2} \vdash C}{
    \ctx , A_1 , A_2 \vdash C}
  \qquad\qquad
  \infer[\lrule{\tensor}]{\ctx , \deli{a{:}A_1 \tensor A_2} \vdash \tensorL{a}{P} :: (c : C)}{
    \ctx , a \pi_1{:}A_1 , a \pi_2{:}A_2 \vdash P :: (c : C)}
\end{equation*}
The $\lrule{\tensor}$ rule demonstrates that, while eligible antecedents always correspond to addresses that are projections $a{\cdot}p$, the converse is not true: not all address projections correspond to eligible antecedents.
Here, although $a \pi_1{:}A_1$ and $a \pi_2{:}A_2$ appear in the premise of the $\lrule{\tensor}$ rule, the antecedents $A_1$ and $A_2$ are \emph{not} eligible in the premise of the $\lrule{\land}$ rule.

(In common layouts we considered, such as the one depicted in the above diagram, the processes $\tensorA{a}$ and $\tensorL{a}{P}$ do not actually write nor read runtime information, their names notwithstanding.
It would be possible to consider erasing $\tensorA{a}$ from the syntax and reducing $\tensorL{a}{P}$ to $P$.
We do not do so because it diverges from the logical foundations and complicates type checking.)

\vspace{\baselineskip}
\noindent
\emph{Example.}
Because SNAX tracks eligibility explicitly, the structure of the proof of commutativity of conjunction changes: $\p{\jrule{SNIP}}$ and $\jrule{ID}$ rules are needed to mediate the ordinary $A_1$ and $A_2$ of the $\lrule{\land}$ rule's premise and the eligible $\eli{A_1}$ and $\eli{A_2}$ of the $\arule{\land}$ axiom.
The corresponding SNAX process involves address projections and explicit copying of data so that the flat layout of pairs is respected.
\begin{equation*}
  \infer[\lrule{\land}]{A_1 \land A_2 \vdash A_2 \land A_1}{
    \infer[\p{\jrule{SNIP}}]{A_1 , A_2 \vdash A_2 \land A_1}{
      \infer[\jrule{ID}]{A_2 \vdash A_2}{} &
      \infer[\p{\jrule{SNIP}}]{\eli{A_2} , A_1 \vdash A_2 \land A_1}{
        \infer[\jrule{ID}]{A_1 \vdash A_1}{} &
        \infer[\arule{\land}]{\eli{A_2} , \eli{A_1} \vdash A_2 \land A_1}{}}}}
  \qquad\quad
  \begin{array}[b]{@{}l@{}l@{}l@{}}
    p : A_1 \tensor A_2 \vdash {} & \tensorL{p}{\\[-1ex]
    & \snip{q \pi_1}{\id{(q \pi_1)}{(p \pi_2)}\,}{\\[-1ex]
    & \snip{q \pi_2}{\id{(q \pi_2)}{(p \pi_1)}\,}{\\[-1ex]
    & \tensorA{q}}}} & {} :: (q : A_2 \tensor A_1)
  \end{array}
\end{equation*}

\paragraph*{Unit, type $\one$.}
The rules for $\top$ in the semi-axiomatic sequent calculus do not involve eligibility -- after all, $\top$ has no proper subformula -- and so the SNAX process constructs and typing rules involving $\one$ are as in SAX.
They are repeated here for convenience.
\begin{equation*}
  \infer[\arule{\one}]{\ctxe \vdash \oneA{a} :: (a : \one)}{}
  \qquad\qquad\quad
  \infer[\lrule{\one}]{\ctx , a{:}\one \vdash \oneL{a}{P} :: (c : C)}{
    \ctx \vdash P :: (c : C)}
\end{equation*}
However, there is one operational difference:
Now that SNAX provides an abstraction that supports conceptually flat layouts of pairs, we can think of data of type $\one$ as taking no space at runtime.
This proves useful in some of the examples that will follow.

\paragraph*{Tagged unions, type $\plus*[\ell \in L]{\ell\colon A_\ell}$.}
Disjunction still corresponds to a labeled sum type, $\plus*[\ell \in L]{\ell\colon A_\ell}$, for tagged unions, as in SAX.
However, instead of the indirect layout of SAX, the presence of an eligible antecedent in the $\arule{\lor}_k$ rule suggests a flat layout for tagged unions in which the tag and the underlying data are laid out side-by-side.
In SNAX, this is abstracted using a new form of address projection, $a \cdot \smash{\ptagg{k}} \vphantom{k}$, where $k$ is the tag.

Instead of SAX's $\plusA[a_k]{a}{k}$ for writing a tagged address, SNAX uses the construct $\plusA{a}{k}$.
This can be seen as fixing address $a_k$ to be $a{\cdot}\smash{\ptagg{k}} \vphantom{k}$ and then eliding it because it is calculable from the address $a$ and tag $k$.
This process is typed by the $\arule{\plus}_k$ rule, which corresponds to the $\arule{\lor}_k$ axiom for disjunction.
Following the pattern seen for pairs, the eligible antecedent $\eli{A_k}$ in the $\arule{\lor}_k$ axiom corresponds to the eligible $\eli{a{\cdot}\smash{\ptagg{k}}{:}A_k} \vphantom{k}$ in the $\arule{\plus}_k$ typing rule.
\begin{equation*}
  \infer[\arule{\lor}_k]{\eli{A_k} \vdash A_1 \lor A_2}{
    (k \in \set{1,2})}
  \qquad\qquad\quad
  \infer[\arule{\plus}_k]{\eli{a{\cdot}\ptagg{k}{:}A_k} \vdash \plusA{a}{k} :: (a : \plus*[\ell \in L]{\ell\colon A_\ell})}{
    (k \in L)}
  \qquad\!\!
  \begin{lgathered}[b]
    \vspace{-2\baselineskip}\\
  \begin{tikzpicture}[baseline=(a.south)]
    \node (a)
      [overlay, draw, label=left:$a$,
       rectangle split, rectangle split parts=4,
       rectangle split empty part width=0.5ex,
       rectangle split empty part height=1.5ex,
       rectangle split horizontal]
      {\nodepart{two}\nodepart{three}${\dotso}$\nodepart{four}};
  \end{tikzpicture}
  \\[-1ex]\qquad\raisebox{1ex}{${}\rightsquigarrow{}$}\!
  \begin{tikzpicture}[baseline=(a.south)]
    \node (a)
      [draw, label=left:$a$,
       rectangle split, rectangle split parts=4,
       rectangle split empty part width=0.5ex,
       rectangle split empty part height=1.5ex,
       rectangle split horizontal]
      {\nodepart{two}\nodepart{three}${\dotso}$\nodepart{four}};
    \node [anchor=base, at=(a.base)] {$k$};
    \node (ak) [overlay, below=2pt of {a.text split |- a.south},
                inner xsep=0pt, anchor=north west] {$a{\cdot}\smash{\ptagg{k}}\vphantom{k}$};
    \node (Ak) [above=0pt of a.north east,
                inner xsep=0pt, anchor=south east] {$\overbrace{\qquad\qquad\enspace}^{\mathclap{\text{data of type $A_k$}}}$};
                \path (ak.north) -- (a.south);
  \end{tikzpicture}
  \end{lgathered}
\end{equation*}

The SNAX construct for reading an address $a$ of type $\plus*[\ell \in L]{\ell\colon A_\ell}$ also differs slightly from its SAX counterpart.
Once again, instead of binding variables $x_\ell$ in the branches of $\plusL[\ell \in L]{a}{[x_\ell]{\ell} => P_\ell}$, we take advantage of knowing that a tag $\ell$ at address $a$ will always have its underlying data located at $a{\cdot}\smash{\ptagg{\ell}} \vphantom{\ell}$ and use the construct $\plusL[\ell \in L]{a}{{\ell} => P_\ell}$.
Logically, the typing rule is still derived from the semi-axiomatic sequent calculus's $\lrule{\lor}$ rule.
All of this follows the pattern seen for the type $A_1 \tensor A_2$.
\begin{equation*}
  \infer[\lrule{\lor}]{\ctx , \deli{A_1 \lor A_2} \vdash C}{
    \forallseq{\ell \in \set{1,2}}{\ctx , A_\ell \vdash C}}
  \qquad\qquad\quad
  \infer[\lrule{\plus}]{\ctx , \deli{a{:}\plus*[\ell \in L]{\ell\colon A_\ell}} \vdash \plusL[\ell \in L]{a}{{\ell} => P_\ell} :: (c : C)}{
    \forallseq{\ell \in L}{\ctx , a {\cdot} \ptagg{\ell}{:}A_\ell \vdash P_\ell :: (c : C)}}
\end{equation*}

\noindent
\emph{Example.}
We can now revisit booleans of type $\bool \defd \plus*{\mt{tt}\colon \one , \mt{ff}\colon \one}$ in the context of SNAX.
Each address $a{:}\bool$ stores only a tag, $\mt{tt}$ or $\mt{ff}$, and no space needs to be reserved for $a{\cdot}\ptag{\mt{tt}}:\one$ or $a{\cdot}\ptag{\mt{ff}}:\one$.
A process for reading a boolean at address $a$ and writing its negation to address $b$ is as follows; its execution in the case that $a$ holds tag $\mt{tt}$ is shown.
\begin{equation*}
  \begin{lgathered}[b]
    \mathrlap{\bool \defd \plus*{\mt{tt}\colon \one , \mt{ff}\colon \one}} \\
    \begin{array}[b]{@{}r@{}l@{}}
      a{:}\bool \vdash {} &
      \mathsf{read}\:a\:(\plusV{\mt{tt}} \Rightarrow \snip{b \cdot \ptag{\mt{ff}}}{\oneA{(b \cdot \ptag{\mt{ff}})}}{\plusA{b}{\mt{ff}}}\\[-1ex]
      & \hphantom{\mathsf{read}\:a\:(}\mathllap{{} \mid {}} \plusV{\mt{ff}} \Rightarrow \snip{b \cdot \ptag{\mt{tt}}}{\oneA{(b \cdot \ptag{\mt{tt}})}}{\plusA{b}{\mt{tt}}}) :: (b : \bool)
    \end{array}
  \end{lgathered}
  \qquad
  \begin{tikzpicture}[baseline=(a.south)]
    \node (a)
      [draw, label=left:$a$,
       rectangle split, rectangle split parts=1,
       rectangle split horizontal,
       rectangle split empty part width=0.5ex,
       rectangle split empty part height=1.5ex]
      {$\mt{tt}\vphantom{\mt{ff}}$};

    \node (b)
      [draw, right=0.65cm of a,
       label=left:$\vphantom{a}\smash{b}$,
       rectangle split, rectangle split parts=1,
       rectangle split horizontal,
       rectangle split empty part width=0.5ex,
       rectangle split empty part height=1.5ex]
      {$\phantom{\mt{f}}\,$};
  \end{tikzpicture}
  \enspace\raisebox{1ex}{${}\rightsquigarrow{}$}\,
  \begin{tikzpicture}[baseline=(a.south)]
    \node (a)
      [draw, label=left:$a$,
       rectangle split, rectangle split parts=1,
       rectangle split horizontal,
       rectangle split empty part width=0.5ex,
       rectangle split empty part height=1.5ex]
      {$\mt{tt}\vphantom{\mt{ff}}$};

    \node (b)
      [draw, right=0.65cm of a,
       label=left:$\smash{b}\vphantom{a}$,
       rectangle split, rectangle split parts=1,
       rectangle split horizontal,
       rectangle split empty part width=0.5ex,
       rectangle split empty part height=1.5ex]
      {$\mt{ff}$};
  \end{tikzpicture}
\end{equation*}

\paragraph*{Pointers, type $\dn A$.}
The semi-axiomatic sequent calculus can include the shift proposition $\dn A$ from adjoint logic~\cite{Benton94csl,Reed09un,Pruiksma18aun}.
Because the semi-axiomatic sequent calculus consists of a single adjoint layer, $A$ and $\dn A$ are logically equivalent.
From a provability perspective, this makes $\dn A$ uninteresting, but the corresponding type, which we also write as $\dn A$, nevertheless has computational significance.\footnote{In future work, we wish to extend the semi-axiomatic sequent calculus and its correspondence with SNAX to have several adjoint layers, to support both linear and persistent data, for example.  In such a system, $\dn A$ would have logical force, as $A$ and $\dn A$ would not be logically equivalent in general.}

As a proposition of positive polarity, $\dn A$ follows the pattern of having its axiom correspond to a typing rule for writes at that type.
\begin{equation*}
  \infer[\arule{\dn}]{\eli{A} \vdash \dn A}{}
  \qquad\qquad\quad
  \infer[\arule{\dn}]{b{:}A \vdash \dnA{a}{b} :: (a : \dn A)}{}
  \qquad
  \begin{tikzpicture}[baseline=(a.base)]
    \node (a)
      [draw, label=left:$a$,
       rectangle split, rectangle split parts=1,
       rectangle split horizontal,
       rectangle split empty part width=0.5ex,
       rectangle split empty part height=1.5ex]
      {\nodepart{two}};
    \node (b) [overlay, above=0.4cm of a.north, text depth=0pt] {$b$};
  \end{tikzpicture}
  \enspace\:\raisebox{1ex}{${}\rightsquigarrow{}$}\,
  \begin{tikzpicture}[baseline=(a.base)]
    \node (a)
      [draw, label=left:$a$,
       rectangle split, rectangle split parts=1,
       rectangle split horizontal,
       rectangle split empty part width=0.5ex,
       rectangle split empty part height=1.5ex]
      {\nodepart{two}};
    \node (b) [overlay, above=0.4cm of a.north, text depth=0pt] {$b$};
    \draw [overlay] (a.north |- a.east) edge[->] (b);
  \end{tikzpicture}
\end{equation*}
From a computational standpoint, we can interpret the structure of the $\arule{\dn}$ axiom as writing an address of type $A$ into an address of type $\dn A$.
In other words, $\dn A$ is the type of pointers to data of type $A$.

In the proof theory, the $\arule{\dn}$ axiom must use an eligible antecedent $\eli{A}$ in order for cut elimination to hold.
However, requiring the $\arule{\dn}$ \emph{typing rule} to use an eligible address would be far too restrictive computationally -- it would force a pointer at address $a$ to point to only a specific projection of $a$, say $a{\cdot}\dn$.
We certainly want pointers to arbitrary addresses, so we allow the $\arule{\dn}$ typing rule to use an ordinary $b{:}A$.
(The $\arule{\dn}$ typing rule is another place the well-formedness condition on typing judgments is essential.)
 
Again following the pattern for positive propositions, the $\lrule{\dn}$ rule corresponds to the typing rule for a construct for reading at type $\dn A$, namely $\dnL{a}{x}{P}$.
\begin{equation*}
  \infer[\lrule{\dn}]{\ctx , \deli{\dn A} \vdash C}{
    \ctx , A \vdash C}
  \qquad\qquad\quad
  \infer[\lrule{\dn}]{\ctx , \deli{a{:}\dn A} \vdash \dnL{a}{x}{P} :: (c : C)}{
    \ctx , x{:}A \vdash P :: (c : C)}
\end{equation*}
Unlike the constructs for reading pairs and tagged values, $\dnL{a}{x}{P}$ has a bound variable and no projection.
At runtime, the address stored at address $a$ is read; then the variable $x$ is bound to that address, and execution continues according to process $P$.
Using a variable, not a projection, in this rule is necessary because the pointer may refer to an arbitrary location, not just to one calculable from $a$.

\vspace{\baselineskip}
\noindent
\emph{Example.}
By judiciously inserting or removing $\dn$ shifts within a type, different layouts can be effected.
As an example, we can revisit the types $\bool$ and $\boollist$.
Having only a $\dn$ shift in front of the recursive call to $\boollist$ yields a flat layout, with pointer indirection only between list elements:
\begin{equation*}
  \begin{lgathered}[b]
    \bool \defd \plus*{\mt{tt}\colon \one , \mt{ff}\colon \one}
    \\
    \boollist \defd \plus*{\mt{nil}\colon \one , \mt{cons}\colon \bool \tensor \dn \boollist}
  \end{lgathered}
  \qquad
  \begin{tikzpicture}[baseline=(xs.base)]
    \node (xs)
      [draw, overlay,
       rectangle split, rectangle split parts=3,
       rectangle split horizontal,
       rectangle split empty part width=0.5ex,
       rectangle split empty part height=1.5ex]
      {$\mt{cons}$\nodepart{two}$\smash{\mt{tt}}\vphantom{\mt{cons}}$\nodepart{three}};
    \path (xs.west) -- (xs.north west);

    \node (xs')
      [overlay, right=0.25cm of xs,
       rectangle split, rectangle split parts=1,
       rectangle split empty part width=0.5ex,
       rectangle split empty part height=1.5ex,
       rectangle split horizontal]
      {${\dotso}$};

    \draw (xs.three south |- xs.east) edge[->] (xs');
  \end{tikzpicture}
\end{equation*}
At the other extreme, having a $\dn$ shift in front of each type constructor effects an indirection-heavy, SAX-like layout within SNAX.
Each $\dn$ shift introduces a pointer into the layout.
\begin{equation*}
  \begin{lgathered}[b]
    \bool \defd \plus*{\mt{tt}\colon \dn \one , \mt{ff}\colon \dn \one}
    \\
    \boollist \defd \plus*{\mt{nil}\colon \dn \one , \mt{cons}\colon \dn (\dn \bool \tensor \dn \boollist)}
  \end{lgathered}
  \qquad
  \begin{tikzpicture}[baseline=(c.base)]
    \node (c)
      [draw, overlay,
       rectangle split, rectangle split parts=2,
       rectangle split horizontal,
       rectangle split empty part width=0.5ex,
       rectangle split empty part height=1.5ex]
      {$\mt{cons}$\nodepart{two}};
    \path (c.west) -- (c.north west);

    \node (p)
      [draw, overlay, right=0.25cm of c,
       rectangle split, rectangle split parts=2,
       rectangle split horizontal,
       rectangle split empty part width=0.5ex,
       rectangle split empty part height=1.5ex]
      {\nodepart{two}};

    \node (xs')
      [overlay, right=0.25cm of p, 
       rectangle split, rectangle split parts=1,
       rectangle split empty part width=0.5ex,
       rectangle split empty part height=1.5ex,
       rectangle split horizontal]
      {${\dotso}$};

    \node (b)
      [draw, overlay, above right=0.25cm and 0cm of p.north west,
       rectangle split, rectangle split parts=2,
       rectangle split empty part width=0.5ex,
       rectangle split empty part height=1.5ex,
       rectangle split horizontal]
      {$\mt{tt}$\nodepart{two}};

    \node (u)
      [overlay, right=0.25cm of b,
       rectangle split, rectangle split parts=1,
       rectangle split empty part width=0.5ex,
       rectangle split empty part height=1.5ex,
       rectangle split horizontal,
       inner xsep=0pt, outer xsep=0pt]
      {$\,\!$};
    \draw [overlay] (u.north west) -- (u.north) -- (u.south) -- (u.south west) -- cycle;

    \draw (c.two south |- c.east) edge[->] (p);
    \draw (p.mid |- p.east) edge[->] (p.mid |- b.south);
    \draw [overlay] (b.two south |- b.east) edge[->] (u);
    \draw (p.two south |- p.east) edge[->] (xs');
  \end{tikzpicture}
\end{equation*}
Layouts with intermediate degrees of indirection can be achieved by using fewer $\dn$ shifts.

\paragraph*{Functions, type $A_1 \imp A_2$.}
Unlike the data of positive types such as $A_1 \tensor A_2$, we will not model the internal layout of function closures.
As previously mentioned, we therefore ignore the eligibility that appears in the semi-axiomatic sequent calculus's $\arule{\limp}$ axiom.
For this reason, SNAX directly inherits the process constructs and static typing rules for $A_1 \imp A_2$ from SAX.
Operationally, they behave as before.
(Although we do not present the details in this paper, it is also possible to adapt negative conjunction from intuitionistic logic to the SNAX type theory as a lazy record type.)
\begin{gather*}
  \infer[\rrule{\limp}]{\ctx \vdash A_1 \limp A_2}{
    \ctx , A_1 \vdash A_2}
  \qquad\quad
  \infer[\rrule{\imp}]{\ctx \vdash \impR{a}{x,z}{P} :: (a : A_1 \imp A_2)}{
    \ctx , x{:}A_1 \vdash P :: (z : A_2)}
  \\[2ex]
  \infer[\arule{\limp}]{A_1 \limp A_2 , \eli{A_1} \vdash \eli{A_2}}{}
  \qquad\quad
  \infer[\arule{\imp}]{a{:}A_1 \imp A_2 , a_1{:}A_1 \vdash \impA{a}{a_1,a_2} :: (a_2 : A_2)}{}
\end{gather*}

\subsection{Adding recursion to SNAX}

Recursion is added to SNAX in the same way as for SAX: We use recursive type and process definitions, $t = A$ and $\fun{\mathit{p}}{(z{:}C)}{(x_1{:}A_1)\dotsm(x_n{:}A_n)} = P$, respectively.
As in SAX, the SNAX type theory requires that these definitions are contractive.
In addition, SNAX requires that all recursion in type definitions be \emph{guarded} by a $\dn$ shift or a negative type constructor (only $\imp$ in this paper), as in $\boollist = \plus*{\mt{nil}\colon \one , \mt{cons}\colon \bool \tensor \dn \boollist}$, for example.
The unguarded type $\boollist = \plus*{\mt{nil}\colon \one , \mt{cons}\colon \bool \tensor \boollist}$ is forbidden because storing a value of type $\boollist$ according to this unguarded definition would require an  unbounded amount of space.

\subsection{Extended example: Mapping a function across a linked list of booleans}

We can revisit the example of mapping a function across a list of booleans in SNAX.
Here we choose to use the flattest of the layouts for lists of booleans, which corresponds to a type definition for $\boollist$ that uses only the $\dn$ shift necessary to guard the recursion.
 A common idiom is for both arguments and destinations of processes to be
addresses (that is, pointers), a small departure from similar code in SAX.
\begin{equation*}
  \!\begin{lgathered}
    \bool \defd \plus*{\mt{tt}\colon \one ,\, \mt{ff}\colon \one}
    \\
    \boollist \defd \plus*{\mt{nil}\colon \one ,\, \mt{cons}\colon \bool \tensor \dn \boollist}
    \\
    \begin{array}[b]{@{}l@{\hspace{1em}}l@{}}
      \fun{\mathit{map}}{(\mathit{ysp} : \dn \boollist)}{(f : \bool \imp \bool)\:(\mathit{xsp} : \dn \boollist)} = {}\\[-1ex]
      \quad\mathsf{read}\:\mathit{xsp}\:(
             \dnV{\mathit{xs}} \Rightarrow {} & \text{\itshape\% read address xs from xsp} \\[-1ex]
      \quad\mathsf{read}\:\mathit{xs}\:( & \text{\itshape\% read and branch on tag at xs} \\[-1ex]
      \quad\hphantom{\mid{}}\plusV{\mt{nil}} \Rightarrow
                   \id{\mathit{ysp}}{\mathit{xsp}} & \text{\itshape\% copy input list pointer xsp to dest.\ ysp} \\[-1ex]\quad
           \mid {} \plusV{\mt{cons}} \Rightarrow
                            \tensorL{(\mathit{xs} \cdot \ptag{cons})}{ & \text{\itshape\% read the pair at $\mathit{xs} \cdot \ptag{cons}$} \\[-1ex]\hphantom{\quad\mid\plusV{\mt{cons}} \Rightarrow {}}
                            \cut{\mathit{ys}}{(
                            \snip{\mathit{ys} \cdot \ptag{cons} \cdot \pi_1}{\impA{f}{\mathit{xs} \cdot \ptag{cons} \cdot \pi_1,\mathit{ys} \cdot \ptag{cons} \cdot \pi_1}\,}{& \text{\itshape\% alloc.\ ys; call f with dest.\ $\mathit{ys} \cdot \ptag{cons} \cdot \pi_1$}\\[-1ex]\hphantom{\quad\mid\plusV{\mt{cons}} \Rightarrow \mathit{ys} \shortleftarrow (}
                            \snip{\mathit{ys} \cdot \ptag{cons} \cdot \pi_2}{
                            \callfun{\mathit{map}}{(\mathit{ys} \cdot \ptag{cons} \cdot \pi_2)}{f\:(\mathit{xs} \cdot \ptag{cons} \cdot \pi_2)}\,}{& \text{\itshape \% call map with dest. $\mathit{ys} \cdot \ptag{cons} \cdot \pi_2$}\\[-1ex]\hphantom{\quad\mid\plusV{\mt{cons}} \Rightarrow \mathit{ys} \shortleftarrow (}
                            \snip{\mathit{ys} \cdot \ptag{cons}}{\tensorA{(\mathit{ys} \cdot \ptag{cons})}\,}{& \text{\itshape\% ``write'' pair to $\mathit{ys} \cdot \ptag{cons}$}\\[-1ex]\hphantom{\quad\mid\plusV{\mt{cons}} \Rightarrow \mathit{ys} \shortleftarrow (}
                            \plusA{\mathit{ys}}{\mt{cons}}}})\,}}{ & \text{\itshape\% write tag $\mt{cons}$ to ys} \\[-1ex]\hphantom{\quad\mid\plusV{\mt{cons}} \Rightarrow {}}
                      \dnA{\mathit{ysp}}{\mathit{ys}}}))} & \text{\itshape\% write pointer to ys at destination ysp}
    \end{array}
  \end{lgathered}
\end{equation*}
After reading the pointer $\mathit{xsp}$ to access $\mathit{xs}{:}\boollist$, the SNAX version of $\mathit{map}$ generally follows the pattern of the SAX $\mathit{map}$, with a few essential deviations.
First, here there is only a single point at which memory is allocated: the cut indicated by the $\mathit{ys} \shortleftarrow ({\dotsm});$ syntax.
Second, the projections such as $\mathit{xs} \cdot \ptag{\mathtt{cons}} \cdot \pi_1$ are used to refer to memory cells within the allocated block as laid out by the type $\boollist$.
Third, because the projections are locally calculable, they can be elided from some parts of the syntax.
Last, because the type now uses $\mathit{ysp} : \dn \boollist$, a final $\dnA{\mathit{ysp}}{\mathit{ys}}$ is needed.

\subsection{Details of the SNAX type theory}

Types and signatures are exactly as they were in SAX, so we do not repeat the details here.

As compared to SAX, SNAX processes have one additional form: $\snip{a}{P}{Q}$ for concurrent composition of processes $P$ and $Q$ that does not allocate memory.
Storables $S$ and co-storables $T$ have slightly different forms than in SAX, owing to SNAX's elision of eligible addresses from process syntax.
\begin{alignat*}{2}
  \text{\emph{Processes}} &\quad&
    P , Q &\Coloneqq \cut{x}{P}{Q} \mid \snip{a}{P}{Q} \mid \id{a}{b} \mid \writeV{a}{S} \mid \callK{a}{T}
  \\
  \text{\emph{Storables}} &&
    S &\Coloneqq \tensorV \mid \oneV \mid \plusV{k} \mid \dnV{a}
            \mid \impK{x,z}{P}
  \\
  \text{\emph{Co-storables}} &&
    T &\Coloneqq \tensorK{P} \mid \oneK{P} \mid \plusK[\ell \in L]{{\ell} => P_\ell} \mid \dnK{x}{P}
            \mid \impV{a_1,a_2}
\end{alignat*}

Once again, we use an operational semantics based on multiset rewriting with semantic objects of the forms $\thread{a}{P}$, $\cell{a}{}$, and $\cell{a}{S}$ that represent running processes, empty cells, and filled cells, respectively.

Configurations $\cnf$ and configuration contexts $\cctx$ are the same as in SAX; once again, configuration contexts are not subject to contraction and their addresses are presumed to be distinct, an invariant that will be preserved by the configuration typing rules.
Also, we will continue to write $a$, $b$, and $c$ in configurations and configuration contexts, but these runtime addresses may not contain static variables $x$.

The configuration typing rules are quite similar to those of SAX, but include one twist.
Unlike SNAX process contexts $\ctx$, configuration contexts $\cctx$ do not track eligibility.
To mediate the two in rules $\jrule{THREAD}$ and $\jrule{CELL}$, we therefore define a judgment $\cctx \vDash_c \ctx$ that holds exactly when three conditions are met:
\begin{enumerate*}[label=\emph{(\roman*)}]
\item $a{:}A \in \ctx$ only if $a{:}A \in \cctx$;
\item $\eli{a{:}A} \in \ctx$ only if $a{:}A \in \cctx$ and $a \extends c$; and 
\item $a{:}A \in \cctx$ and $a \extends c$ only if either $\eli{a{:}A} \in \ctx$ or $a \extends b$ for some $\eli{b{:}B} \in \ctx$.
\end{enumerate*}
The process typing premises in rules $\jrule{THREAD}$ and $\jrule{CELL}$ then further guarantee that the choice of eligible antecedents is consistent with the process. 
\begin{mathpar}
  \infer[\jrule{EMP}]{\cctx \vDash \cnfe* :: \cctx}{}
  \and
  \infer[\jrule{JOIN}]{\cctx \vDash \cnf_1 \cc \cnf_2 :: \cctx''}{
    \cctx \vDash \cnf_1 :: \cctx' & 
    \cctx' \vDash \cnf_2 :: \cctx''}
  \\
  \infer[\jrule{THREAD}]{\cctx \vDash \thread{a}{P} \cc \cell{a}{} :: (\cctx , a{:}A)}{
    (a \notin \dom{\cctx}) &
    \cctx \vDash_a \ctx &
    \ctx \vdash P :: (a : A)}
  \and
  \infer[\jrule{CELL}]{\cctx \vDash \cell{a}{S} :: (\cctx , a{:}A)}{
    (a \notin \dom{\cctx}) &
    \cctx \vDash_a \ctx &
    \ctx \vdash \writeV{a}{S} :: (a : A)}
\end{mathpar}
The rewriting rules for futures, writes, reads, and calls are essentially the same as in SAX, but the definition of $\pass{S}{T}$ changes slightly.
Because some addresses are locally calculable and elided from the syntax, substitution is no longer needed in some cases; however, substitution is still needed for the cases for pointers and functions.
\begin{equation*}
  \begin{lgathered}[t]
    \begin{aligned}[t]
      \MoveEqLeft[.5]
      \thread{c}{(\cut{x}{P}{Q})}\qquad\qquad\qquad\qquad\,\text{($\alpha$ fresh)} \\
        &\stepsto \thread[\alpha/x]{\alpha}{P} \cc \cell{\alpha}{} \cc \thread[\alpha/x]{c}{Q}
    \end{aligned}
    \\[0.5ex]
    \thread{a}{\writeV{a}{S}} \cc \cell{a}{}
      \stepsto \cell{a}{S}
    \\[0.5ex]
    \thread{c}{\readV{a}{T}} \cc \cell{a}{S}
      \stepsto \thread{c}{\pass{S}{T}}
    \\[0.5ex]
    \begin{aligned}[t]
      \MoveEqLeft[.5]
      \thread{c}{\callfun{\mathit{p}}{c}{a_1 \dotsb\mkern1.5mu a_n}} \\
        &\stepsto \thread{c}{[c/z, a_1/x_1, \dotsc, a_n/x_n]P} \\
        &\mathrel{\text{(where $\fun{\mathit{p}}{z}{x_1 \dotsb\mkern1.5mu x_n} = P$)}} {}
    \end{aligned}
  \end{lgathered}
  \qquad
  \begin{lgathered}[t]
    \text{where $\pass{S}{T}$ is given by} \\[0.5ex]\quad
    \begin{aligned}[t]
      \pass{\tensorV}{\tensorK{P}} &= P \\[0.5ex]
      \pass{\oneV}{\oneK{P}} &= P \\[0.5ex]
      \pass{\plusV{k}}{\plusK[\ell \in L]{{\ell} => P_\ell}} &= P_k \quad (k \in L) \\[0.5ex]
      \pass{\ptrV{a}}{\ptrK{x}{P}} &= [a/x]P \\[0.5ex]
      \pass{\impV{a_1,a_2}}{\impK{x,z}{P}} &= [a_1/x,a_2/z]P
    \end{aligned}
  \end{lgathered}
\end{equation*}

Snips are the essential difference between SAX and SNAX.
The rewriting rule for a snip is broadly similar to that for futures, with the key difference that an address $a$ is used instead of choosing a fresh $\alpha$.
\begin{equation*}
  \thread{c}{(\snip{a}{P}{Q})}
    \stepsto \thread{a}{P} \cc \cell{a}{} \cc \thread{c}{Q}
    \text{ if $\cseq{P} = \set{a}$}
\end{equation*}
Because the address $a$ written by $P$ is not made locally explicit in the snip construct $\snip{a}{P}{Q}$, the function $\cseq{P}$ (short for ``destination'') traverses the process $P$ to extract that address.
This function returns a set of addresses, but for a well-typed process $P$, the set will always be a singleton.
The definition of $\cseq{P}$ can be found in \cref{app:cseq}.

Copying data must be handled differently in SNAX than in SAX.
In SAX, we could simply copy a storable from one address to another; the implicit sharing would take care of a type's subformulas.
In SNAX, all sharing is made explicit through the $\dn$ shifts, and those may or may not appear in a given type's subformulas.
However, at types $\dn A$ and $A_1 \imp A_2$, simply copying the storable suffices.
\begin{equation*}
  \thread{a}{\id{a}{b}} \cc \cell{b}{S}
    \stepsto \cell{a}{S}
\end{equation*}
Before a process may be executed, we require that $\mathsf{copy}$s at other types are expanded, using reads and writes, so that only $\mathsf{copy}$s at types $\dn A$ and $A_1 \imp A_2$ remain; this is reminiscent of $\eta$-expansion.
The expansion of $\mathsf{copy}$s is shown in \cref{app:cseq}.

\subsection{Type safety for SNAX}\label{sec:SNAX:safety}

SNAX satisfies type safety, in the form of type preservation and progress results.
Preservation is a bit subtle, but ultimately not difficult, to prove; it relies on various lemmas surrounding eligibility, as well as the definition of $\cctx \vDash_c \ctx$.
\begin{theorem}[Preservation]
  If $\cctx_0 \vDash \cnf :: \cctx$ and $\cnf \stepsto \cnf'$, then $\cctx_0 \vDash \cnf' :: \cctx'$ for some $\cctx' \supseteq \cctx$.
\end{theorem}
\begin{proof}
  By induction on  the given derivation, using a few lemmas about eligibility; see \cref{app:metatheorems}.
\end{proof}

\begin{theorem}[Progress]
  If\/ $\vDash \cnf :: \cctx$, then either $\cnf$ is final or $\cnf \stepsto \cnf'$ for some $\cnf'$.
\end{theorem}
\begin{proof}
  By right-to-left induction on the structure of the given derivation; see \cref{app:metatheorems}.
\end{proof}

\section{Related work}\label{sec:related}

Besides the aforementioned work on the semi-axiomatic sequent calculus and SAX~\cite{DeYoung20fscd}, another item of related work is Smullyan's classical sequent calculus in which cuts must be analytic and all other inference rules are replaced by axioms~\cite{Smullyan68jsl}.
Because all cuts are analytic, there is no direct cut elimination procedure and, consequently, the calculus does not seem to lend itself to computational interpretation.

From a computational standpoint, most closely related is perhaps the work on data layout using
\emph{ordered types}~\cite{Petersen03popl}.  Ordered types were suitable to capture the original
allocation and layout of data, but not the whole state of memory during computation
since ordered logic has only a single ordered
context thus cannot directly model many blocks of memory connected by pointers.
The current design
overcomes both of these limitations with a very different approach: our logic (and therefore the
type theory) is not substructural at all.

Another point of comparison is Typed Assembly Language
(TAL)~\cite{Morrisett99toplas}.  We view TAL as a low-level type system that can reflect high level abstractions,
but it does not seem to correspond to any particular proof system for intuitionistic logic.
Furthermore, while TAL by necessity works with concrete data layouts, the compilation from the
$\lambda$-calculus to TAL chooses a particular one among them rather than providing a choice to the
programmer.  Another point of difference is that in SNAX, functions receive \emph{destinations}
(that is, memory locations) for their results, while in TAL they receive \emph{continuations} to be
called with the result.  TAL also resolves some issues that we leave to future work.  Among them
are parametric polymorphism and representation of closures.

\section{Conclusion}

We have shown how elevating notions of eligibility and snips that arise in the semi-axiomatic sequent calculus's cut elimination proof from refinement properties to first-class logical concepts yields a Curry--Howard explanation of (abstract) data layout in futures-based shared memory concurrency.
Moreover, we have proved type preservation and progress for the resulting SNAX type theory.

In future work, we plan to extend SNAX to support parametric polymorphism, as well as adjoint layers for integrating a treatment of linear data with SNAX's existing treatment of persistent, write-once data.
In designing both extensions, we will be able to lean on SNAX's strong logical foundations.
Studying code optimization in the SNAX setting is another avenue for future work that we are pursuing.

\bibliographystyle{../../entics}
\bibliography{mfps22}

\appendix
\section{Auxiliary definitions for SNAX operational semantics}\label{app:cseq}

The definition of $\cseq{P}$, which is used in the operational semantics of snips, is as follows.
\begin{align*}
  \cseq{\cut{x}{P}{Q}} &= \cseq{Q} \\
  \cseq{\snip{a}{P}{Q}} &= \cseq{Q} \\
  \cseq{\id{a}{b}} &= \set{a} \\
  \cseq{\writeV{a}{S}} &= \set{a} \\
  \cseq{\tensorL{a}{P}} &= \cseq{P} \\
  \cseq{\oneL{a}{P}} &= \cseq{P} \\
  \cseq{\plusL[\ell \in L]{a}{{\ell} => P_\ell}} &= {\textstyle \bigcup_{\ell \in L} \cseq{P_\ell}} \\
  \cseq{\dnL{a}{x}{P}} &= \cseq{P} \\
  \cseq{\impA{a}{a_1,a_2}} &= \set{a_2}
\end{align*}
Expansion of $\mathsf{copy}$s down to types $\dn A$ and $A_1 \imp A_2$ is accomplished by the following function.
\begin{align*}
  \eta(\id{a}{b} : A_1 \tensor A_2)
    &= \tensorL{b}{\snip{a \pi_1}{\eta(\id{(a \pi_1)}{(b \pi_1)} : A_1)}{\snip{a \pi_2}{\eta(\id{(a \pi_2)}{(b \pi_2)} : A_2)}{\tensorA{a}}}}
  \\
  \eta(\id{a}{b} : \one)
    &= \oneL{b}{\oneA{a}}
  \\
  \eta(\id{a}{b} : \plus*[\ell \in L]{\ell: A_\ell})
    &= \plusL[\ell \in L]{b}{{\ell} => \snip{a \cdot \ptagg{\ell}}{\eta(\id{(a{\cdot}\ptagg{\ell})}{(b{\cdot}\ptagg{\ell})} : A_\ell)}{\plusA{a}{\ell}}}
  \\
  \eta(\id{a}{b} : \dn A)
    &= \id{a}{b}
  \\
  \eta(\id{a}{b} : A \imp B)
    &= \id{a}{b}
\end{align*}

\section{SNAX Metatheorems}\label{app:metatheorems}

\begin{lemma}\label{lem:suffix}
  If $a \extends b$ and $a \extends c$, then either $b \wextends c$ or $c \wextends b$.
\end{lemma}
\begin{proof}
  By proving by simultaneous induction on $p_1$ and $p_2$ that $b \cdot p_1 = c \cdot p_2$ implies $b \wextends c$ or $c \wextends b$.
\end{proof}

\begin{lemma}\label{lem:extends}
  If $\ctx , \eli{a{:}A} \vdash P :: (c : C)$, then $a \extends c$.
\end{lemma}
\begin{proof}
  By induction on the structure of the given derivation.
  The two interesting cases are as follows.

  \paragraph*{Case:}
  \begin{equation*}
    \infer[\p{\jrule{SNIP}}]{\ctx_1 , \ctx_2 , \eli{a{:}A} \vdash \snip{b}{P}{Q} :: (c : C)}{
      \ctx_1 , \eli{a{:}A} \vdash P :: (b : B) &
      \ctx_2 , \eli{b{:}B} \vdash Q :: (c : C)}
  \end{equation*}
  Appealing to the inductive hypothesis on the first premise, we know that $a \extends b$.
  Similarly, appealing to the inductive hypothesis on the second premise, we know that $b \extends c$.
  So $a \extends c$ follows from transitivity.

  \paragraph*{Case:}
  \begin{equation*}
    \infer[\jrule{CUT}]{\ctx_1 , \ctx_2 , \eli{a{:}A} \vdash \cut{x}{P}{Q} :: (c : C)}{
      \ctx_1 , \eli{a{:}A} \vdash P :: (x : B) &
      \ctx_2 , x{:}B \vdash Q :: (c : C) &
      \text{($x$ fresh)}}
  \end{equation*}
  By the inductive hypothesis on the first premise, we know that $a \extends x$.
  However, then having $\eli{a{:}A}$ in the rule's conclusion contradicts the freshness of $x$.
  (The case for the $\rrule{\imp}$ rule is similar.)
\end{proof}

\begin{lemma}\label{lem:twoeli}
  If $\ctx , \eli{a{:}A} , \eli{b{:}B} \vdash P :: (c : C)$, then $a \nwextends b$ and $b \nwextends a$.
\end{lemma}
\begin{proof}
  By induction on the structure of the given derivation.
  The two interesting cases are as follows.
  \paragraph*{Case:}
  \begin{equation*}
    \infer[\p{\jrule{SNIP}}]{\ctx_1 , \ctx_2 , \eli{a{:}A} , \eli{b{:}B} \vdash \snip{c'}{P}{Q} :: (c : C)}{
      \ctx_1 , \eli{a{:}A} \vdash P :: (c' : C') &
      \ctx_2 , \eli{b{:}B} , \eli{c'{:}C'} \vdash Q :: (c : C)}
  \end{equation*}
  We must show that $a \neq b$ and $a \nextends b$ and $b \nextends a$.
  \begin{itemize}
  \item Suppose that $a = b$.
    We know from the first premise above and Lemma~\ref{lem:extends} that $a \extends c'$.
    So $b \extends c'$ as well.
    By the inductive hypothesis on the second premise above, $b \nwextends c'$, yielding a contradiction.
    Therefore $a \neq b$.
  \item Suppose that $a \extends b$.
    Once again, we know from the first premise above and Lemma~\ref{lem:extends} that $a \extends c'$.
    By Lemma~\ref{lem:suffix}, either $b \wextends c'$ or $c' \wextends b$.
    Appealing to the inductive hypothesis on the second premise above, $b \nwextends c'$ and $c' \nwextends b$.
    This is a contradiction, so $a \nextends b$.
  \item Suppose that $b \extends a$.
    Once again, we know from the first premise above and Lemma~\ref{lem:extends} that $a \extends c'$.
    So $b \extends c'$ follows by transitivity of $\extends$.
    Appealing to the inductive hypothesis on the second premise above, $b \nwextends c'$ and $c' \nwextends b$.
    This is a contradiction, so $b \nextends a$.
  \end{itemize}

  \paragraph*{Case:}
  \begin{equation*}
    \infer[\jrule{CUT}]{\ctx_1 , \ctx_2 , \eli{a{:}A} , \eli{b{:}B} \vdash \cut{x}{P}{Q} :: (c : C)}{
      \ctx_1 , \eli{a{:}A} \vdash P :: (x : C') &
      \ctx_2 , \eli{b{:}B} , x{:}C' \vdash Q :: (c : C) &
      \text{($x$ fresh)}}
  \end{equation*}
  By Lemma~\ref{lem:extends} on the first premise, we know that $a \extends x$.
  However, then having $\eli{a{:}A}$ in the rule's conclusion contradicts the freshness of $x$.
  (The case for the $\rrule{\imp}$ rule is similar.)
\end{proof}

\begin{lemma}\label{lem:modelsweaken}
  If $\cctx \vDash_c \ctx , a{:}A$, then $\cctx \vDash_c \ctx$.
\end{lemma}
\begin{proof}
  Assume $\cctx \vDash_c \ctx , a{:}A$.
  To establish $\cctx \vDash_c \ctx$, there are three parts.
  \begin{itemize}
  \item Assume $b{:}B \in \ctx$.
    Then $b{:}B \in \ctx , a{:}A$ as well.
    It follows from $\cctx \vDash_c \ctx , a{:}A$ that $b{:}B \in \cctx$.
  \item Assume that $\eli{b{:}B} \in \ctx$.
    Then $\eli{b{:}B} \in \ctx , a{:}A$ as well.
    It follows from $\cctx \vDash_c \ctx , a{:}A$ that $b{:}B \in \cctx$ and $b \extends c$.
  \item Assume that $b{:}B \in \cctx$ and $b \extends c$.
    It follows from $\cctx \vDash_c \ctx , a{:}A$ that either $\eli{b{:}B} \in \ctx , a{:}A$ or $b \extends b'$ for some $\eli{b'{:}B'} \in \ctx , a{:}A$.
    Because $a{:}A$ is ordinary, either $\eli{b{:}B} \in \ctx$ or $b \extends b'$ for some $\eli{b'{:}B'} \in \ctx$.
  \end{itemize}
\end{proof}

\begin{theorem}[Preservation]
  If $\cctx_0 \vDash \cnf :: \cctx$ and $\cnf \stepsto \cnf'$, then $\cctx_0 \vDash \cnf' :: \cctx'$ for some $\cctx' \supseteq \cctx$.
\end{theorem}
\begin{proof}
  By induction on the structure of the given derivation, appealing to the preceding lemmas about eligibility.
  The most interesting case is as follows.
  \paragraph*{Case:}
  \begin{gather*}
    \infer[\jrule{THREAD}]{\cctx \vDash \thread{c}{\snip{a}{P}{Q}} \cc \cell{c}{} :: (\cctx , c{:}C)}{
      (c \notin \dom{\cctx}) &
      \cctx \vDash_c \ctx_1 , \ctx_2 &
      \infer[\p{\jrule{SNIP}}]{\ctx_1 , \ctx_2 \vdash \snip{a}{P}{Q} :: (c : C)}{
        \ctx_1 \vdash P :: (a : A) &
        \ctx_2 , \eli{a{:}A} \vdash Q :: (c : C)}}
    \\\stepsto\\
    \infer[\jrule{JOIN}]{\cctx \vDash \thread{a}{P} \cc \cell{a}{} \cc \thread{c}{Q} \cc \cell{c}{} :: (\cctx , a{:}A , c{:}C)}{
    \infer{\cctx \vDash \thread{a}{P} \cc \cell{a}{} :: (\cctx , a{:}A)}{
      (a \notin \dom{\cctx}) &
      \cctx \vDash_a \ctx_1 &
      \ctx_1 \vdash P :: (a : A)} &
    \infer{\cctx , a{:}A \vDash \thread{c}{Q} \cc \cell{c}{} :: (\cctx , a{:}A , c{:}C)}{
      (c \notin \dom{(\cctx , a{:}A)}) &
      \cctx , a{:}A \vDash_c \ctx_2 , \eli{a{:}A} &
      \ctx_2 , \eli{a{:}A} \vdash Q :: (c : C)}}
  \end{gather*}
  First, we must show that $a \notin \dom{\cctx}$.
    \begin{itemize}
    \item Suppose that $a \in \dom{\cctx}$.
      From the snip's second premise, we know that $a \extends c$ (Lemma~\ref{lem:extends}).
      Because $\cctx \vDash_c \ctx_1 , \ctx_2$, either: $\eli{a{:}A} \in \ctx_1$; $\eli{a{:}A} \in \ctx_2$; or $a \extends b$ for some $\eli{b{:}B} \in \ctx_1 , \ctx_2$.
      \begin{itemize}
      \item If $\eli{a{:}A} \in \ctx_1$, then Lemma~\ref{lem:extends} on the first premise yields $a \extends a$, which is impossible.
      \item If $\eli{a{:}A} \in \ctx_2$, then Lemma~\ref{lem:twoeli} on the second premise yields $a \nwextends a$, which is impossible.
      \item Otherwise, $a \extends b$ for some $\eli{b{:}B} \in \ctx_1 , \ctx_2$.
        If $\eli{b{:}B} \in \ctx_1$, then Lemma~\ref{lem:extends} yields $b \extends a$, which contradicts $a \extends b$.
        If $\eli{b{:}B} \in \ctx_2$, then Lemma~\ref{lem:twoeli} yields $a \nwextends b$, which contradicts $a \extends b$.
      \end{itemize}
    \end{itemize}
  Second, we must show that $\cctx \vDash_a \ctx_1$.
    \begin{itemize}
    \item Assume that $b{:}B \in \ctx_1$.
      From $\cctx \vDash_c \ctx_1 , \ctx_2$, we therefore know that $b{:}B \in \cctx$, as required.

    \item Assume that $\eli{b{:}B} \in \ctx_1$.
      From $\cctx \vDash_c \ctx_1 , \ctx_2$, we therefore know that $b{:}B \in \cctx$ (and $b \extends c$).
      Because $\eli{b{:}B} \in \ctx_1$, Lemma~\ref{lem:extends} on the first premise yields $b \extends a$, as required.
    \item Assume that $b{:}B \in \cctx$ and $b \extends a$.
      From $\cctx \vDash_c \ctx_1 , \ctx_2$, we therefore know that either $\eli{b{:}B} \in \ctx_1 , \ctx_2$ or $b \extends b'$ for some $\eli{b'{:}B'} \in \ctx_1 , \ctx_2$.
      \begin{itemize}
      \item Suppose that $\eli{b{:}B} \in \ctx_2$.
        By Lemma~\ref{lem:twoeli} on the second premise, $b \nwextends a$, which contradicts $b \extends a$.
      \item Suppose that $b \extends b'$ for some $\eli{b'{:}B'} \in \ctx_2$.
        Because both $b \extends a$ and $b \extends b'$, Lemma~\ref{lem:suffix} yields either $a \wextends b'$ or $b' \wextends a$.
        However, by Lemma~\ref{lem:twoeli} and the second premise, neither of these can be true.
      \end{itemize}
      The only remaining possibility is that either $\eli{b{:}B} \in \ctx_1$ or $b \extends b'$ for some $\eli{b'{:}B'} \in \ctx_1$, as required.
    \end{itemize}
  Third, we must show that $c \notin \dom{(\cctx , a{:}A)}$.
  \begin{itemize}
  \item We are given that $c \notin \dom{\cctx}$.
    From Lemma~\ref{lem:extends} and the snip's second premise, we know that $a \extends c$.
    This also implies that $c \neq a$, so we may indeed conclude that $c \notin \dom{(\cctx , a{:}A)}$.
  \end{itemize}
  Fourth, we must show that $\cctx , a{:}A \vDash_c \ctx_2 , \eli{a{:}A}$.
  \begin{itemize}
  \item Assume that $b{:}B \in \ctx_2 , \eli{a{:}A}$.
    More precisely, $b{:}B \in \ctx_2$.
    Because $\cctx \vDash_c \ctx_1 , \ctx_2$, it follows that $b{:}B \in \cctx$.
  \item Assume that $\eli{b{:}B} \in \ctx_2 , \eli{a{:}A}$.
    \begin{itemize}
    \item If $\eli{b{:}B} \in \ctx_2$, then it follows from $\cctx \vDash_c \ctx_1 , \ctx_2$, it follows that $b{:}B \in \cctx$ and $b \extends c$, as required.
    \item Otherwise, $b = a$ and $B = A$.
      Then $b{:}B \in \cctx , a{:}A$.
      Also, by Lemma~\ref{lem:extends} on the first premise, $a \extends c$.
      So $b \extends c$, as required.
    \end{itemize}
    \item Assume that $b{:}B \in \cctx , a{:}A$ and $b \extends c$.
      If $b = a$ and $B = A$, then $\eli{b{:}B} \in \ctx_2 , \eli{a{:}A}$.
      Otherwise, $b{:}B \in \cctx$.
      From $\cctx \vDash_c \ctx_1 , \ctx_2$, we therefore know that either $\eli{b{:}B} \in \ctx_1 , \ctx_2$ or $b \extends b'$ for some $\eli{b'{:}B'} \in \ctx_1 , \ctx_2$.
      \begin{itemize}
      \item Suppose that $\eli{b{:}B} \in \ctx_1$.
        By Lemma~\ref{lem:extends} on the first premise, $b \extends a$.
        And $\eli{a{:}A} \in \ctx_2 , \eli{a{:}A}$, as required.
      \item Suppose that $b \extends b'$ for some $\eli{b'{:}B'} \in \ctx_1$.
        By Lemma~\ref{lem:extends} and the first premise, $b' \extends a$.
        By transitivity, $b \extends a$.
        And $\eli{a{:}A} \in \ctx_2 , \eli{a{:}A}$, as required.
      \end{itemize}
      Therefore, in all cases, either $\eli{b{:}B} \in \ctx_2 , \eli{a{:}A}$ or $b \extends b'$ for some $\eli{b'{:}B'} \in \ctx_2 , \eli{a{:}A}$, as required.
    \end{itemize}

  \paragraph*{Case:}
  \begin{gather*}
    \infer[\jrule{THREAD}]{\cctx_0 \vDash \thread{c}{(\cut{x}{P}{Q})} \cc \cell{c}{} :: (\cctx_0 , c{:}C)}{
      (c \notin \dom{\cctx_0}) &
      \cctx_0 \vDash_c \ctx_1 , \ctx_2 &
      \infer[\jrule{CUT}]{\ctx_1 , \ctx_2 \vdash \cut{x}{P}{Q} :: (c : C)}{
        \ctx_1 \vdash P :: (x : A) &
        \ctx_2 , x{:}A \vdash Q :: (c : C) &
        \text{($x$ fresh)}}}
    \\\stepsto\\
    \infer{\cctx_0 \vDash \thread[\alpha/x]{\alpha}{P} \cc \cell{\alpha}{} \cc \thread[\alpha/x]{c}{Q} \cc \cell{c}{} :: (\cctx_0 , \alpha{:}A , c{:}C)}{
      \deduce{\cctx_0 \vDash \thread[\alpha/x]{\alpha}{P} \cc \cell{\alpha}{} :: (\cctx_0 , \alpha{:}A)}{\mathcal{D}} &
      \deduce{\cctx_0 , \alpha{:}A \vDash \thread[\alpha/x]{c}{Q} \cc \cell{c}{} :: (\cctx_0 , \alpha{:}A , c{:}C)}{\mathcal{E}}}
  \end{gather*}
  where
  \begin{gather*}
    \mathcal{D}
    =
    \infer{\cctx_0 \vDash \thread[\alpha/x]{\alpha}{P} \cc \cell{\alpha}{} :: (\cctx_0 , \alpha{:}A)}{
      (\alpha \notin \dom{\cctx_0}) &
      \cctx_0 \vDash_\alpha \ctx_1 &
      \infer-{\ctx_1 \vdash [\alpha/x]P :: (\alpha : A)}{
        \ctx_1 \vdash P :: (x : A)}}
  \end{gather*}
  and
  \begin{gather*}
    \mathcal{E}
    =
    \infer{\cctx_0 , \alpha{:}A \vDash \thread[\alpha/x]{c}{Q} \cc \cell{c}{} :: (\cctx_0 , \alpha{:}A , c{:}C)}{
      (c \notin \dom{(\cctx_0 , \alpha{:}A)}) &
      \infer{\cctx_0 , \alpha{:}A \vDash_c \ctx_2 , \alpha{:}A}{
        \infer{\cctx_0 , \alpha{:}A \vDash_c \ctx_2}{
          \cctx_0 \vDash_c \ctx_2}} &
      \infer-{\ctx_2 , \alpha{:}A \vdash [\alpha/x]Q :: (c : C)}{
        \ctx_2 , x{:}A \vdash Q :: (c : C)}}
  \end{gather*}
  \begin{itemize}
  \item $\alpha \notin \dom{\cctx_0}$ because $\alpha$ is chosen to be fresh.
  \item Because $c \notin \dom{\cctx_0}$ is given and $\alpha$ is fresh, $c \notin \dom{(\cctx_0 , \alpha{:}A)}$ follows.
  \item Because $\ctx_1 \vdash P :: (x : A)$ for a fresh $x$, the context $\ctx_1$ must not contain any eligible addresses.
    Fortunately, because $\alpha$ is fresh, no address in $\dom{\cctx_0}$ will have the form $\alpha \cdot p$.
    Therefore $\cctx_0 \vDash_\alpha \ctx_1$.

    Moreover, because $\ctx_1$ contains no eligible addresses, a lemma gives $\cctx_0 \vDash_c \ctx_2$ from $\cctx_0 \vDash_c \ctx_1 , \ctx_2$.
  \end{itemize}
\end{proof}

\begin{lemma}\label{lem:cseq}
  If\/ $\ctx \vdash P :: (a : A)$, then $\cseq{P} = \set{a}$.
\end{lemma}
\begin{proof}
  By induction on the structure of the given derivation.
  The most interesting case is as follows.
  \paragraph*{Case:}
  \begin{equation*}
    \infer[\lrule{\plus}]{\ctx , a{:}\plus*[\ell \in L]{\ell\colon A_\ell} \vdash \plusL[\ell \in L]{a}{{\ell} => P_\ell} :: (c : C)}{
      \forallseq{\ell \in L}{\ctx , a{\cdot}\ptagg{\ell}{:}A_\ell \vdash P_\ell :: (c : C)}}
  \end{equation*}
  By the inductive hypothesis, $\cseq{P_\ell} = \set{c}$ for all $\ell \in L$.
  Then $\bigcup_{\ell \in L} \cseq{P_\ell} = \set{c}$, as required.
\end{proof}

\begin{theorem}[Progress]
  If\/ $\vDash \cnf :: \cctx$, then either $\cnf$ is final or $\cnf \stepsto \cnf'$ for some $\cnf'$.
\end{theorem}
\begin{proof}
  By right-to-left induction on the structure of the given derivation.
  \paragraph*{Case:}
  \begin{equation*}
    \infer[\jrule{JOIN}]{\vDash \cnf \cc \thread{c}{(\snip{a}{P}{Q})} \cc \cell{c}{} :: (\cctx , c{:}C)}{
      \vDash \cnf :: \cctx &
      \infer{\cctx \vDash \thread{c}{(\snip{a}{P}{Q})} \cc \cell{c}{} :: (\cctx , c{:}C)}{
        (c \notin \dom{\cctx}) &
        \cctx \vDash_c \ctx_1 , \ctx_2 &
        \infer[\p{\jrule{SNIP}}]{\ctx_1 , \ctx_2 \vdash \snip{a}{P}{Q} :: (c : C)}{
          \ctx_1 \vdash P :: (a : A) & \ctx_2 , \eli{a{:}A} \vdash Q :: (c : C)}}}
  \end{equation*}
  By Lemma~\ref{lem:cseq}, we have $\cseq{P} = \set{a}$.
  Therefore,
  \begin{equation*}
    \cnf \cc \thread{c}{(\snip{a}{P}{Q})} \cc \cell{c}{} \stepsto \cnf \cc \thread{a}{P} \cc \cell{a}{} \cc \thread{c}{Q} \cc \cell{c}{}
    \,,
  \end{equation*}
  as required.

  \paragraph*{Case:}
  \begin{equation*}
    \infer[\jrule{JOIN}]{\vDash \cnf \cc \thread{c}{\plusL[\ell \in L]{a}{{\ell} => P_\ell}} \cc \cell{c}{} :: (\cctx , c{:}C)}{
      \vDash \cnf :: \cctx &
      \infer{\cctx \vDash \thread{c}{\plusL[\ell \in L]{a}{{\ell} => P_\ell}} \cc \cell{c}{} :: (\cctx , c{:}C)}{
        (c \notin \dom{\cctx}) &
        \cctx \vDash_c \ctx , a{:}\plus*[\ell \in L]{\ell\colon A_\ell} &
        \infer[\lrule{\plus}]{\ctx , a{:}\plus*[\ell \in L]{\ell\colon A_\ell} \vdash \plusL[\ell \in L]{a}{{\ell} => P_\ell} :: (c : C)}{
          \forallseq{\ell \in L}{\ctx , a{\cdot}\ptagg{\ell}{:}A_\ell \vdash P_\ell :: (c : C)}}}}
  \end{equation*}
  By the inductive hypothesis, either $\cnf$ is final or $\cnf \stepsto \cnf'$ for some $\cnf'$.
  \begin{itemize}
  \item Suppose that $\cnf \stepsto \cnf'$ for some $\cnf'$.
    Then $\cnf \cc \thread{c}{\plusL[\ell \in L]{a}{{\ell} => P_\ell}} \cc \cell{c}{} \stepsto \cnf' \cc \thread{c}{\plusL[\ell \in L]{a}{{\ell} => P_\ell}} \cc \cell{c}{}$.
  \item Suppose that $\cnf$ is final.
    Because $\cctx \vDash_c \ctx , a{:}\plus*[\ell \in L]{\ell\colon A_\ell}$, we also have $a{:}\plus*[\ell \in L]{\ell\colon A_\ell} \in \cctx$.
    Then, by inversion on the derivation of $\vDash \cnf :: \cctx$, it follows that $a{\cdot}\ptagg{k}{:}A_k \in \cctx$ and $\cnf = \cnf_0 \cc \cell{a{\cdot}\ptagg{k}}{\plusV{k}}$ for some $\cnf_0$ and $k \in L$.
    Therefore $\cnf \cc \thread{c}{\plusL[\ell \in L]{a}{{\ell} => P_\ell}} \cc \cell{c}{} \stepsto \cnf \cc \thread{c}{P_k} \cc \cell{c}{}$, as required.
  \end{itemize}
\end{proof}

\end{document}